\documentclass[11pt]{article}

\usepackage{fullpage}
\usepackage[T1]{fontenc}
\usepackage{mathpazo}
\usepackage{microtype}
\usepackage{needspace}
\usepackage{amsmath,amssymb,amsfonts,amsthm,mathtools}
\usepackage{mathrsfs}
\usepackage{bm}
\usepackage{enumitem}
\usepackage{booktabs}
\usepackage[dvipsnames,svgnames,x11names]{xcolor}
\usepackage{hyperref}
\usepackage[nameinlink,capitalize,noabbrev]{cleveref}
\usepackage{todonotes}

\definecolor{DarkRed}{RGB}{170,0,0}
\hypersetup{
  pdftitle={When Unitary Designs Are Not Pseudorandom},
  pdfauthor={Jesko Dujmovic, Jonas Haferkamp, Alexander Poremba},
  colorlinks=true,
  linkcolor=DarkRed,
  citecolor=DarkRed,
  urlcolor=DarkRed,
  breaklinks=true
}

\usepackage{cleveref} 

\usepackage{tikz}
\usetikzlibrary{arrows.meta}

\usepackage{hyperref}
\usepackage{aliascnt}                 
\usepackage[nameinlink,capitalize,noabbrev]{cleveref}

\newtheorem{theorem}{Theorem}[section]

\newaliascnt{lemma}{theorem}
\newtheorem{lemma}[lemma]{Lemma}
\aliascntresetthe{lemma}

\newtheorem{proposition}[theorem]{Proposition}
\newtheorem{definition}[theorem]{Definition}
\newtheorem{conjecture}[theorem]{Conjecture}

\newcommand{\C}{\mathbb C}
\newcommand{\F}{\mathbb F}
\newcommand{\E}{\mathbb E}

\newcommand{\poly}{\operatorname{poly}}
\newcommand{\negl}{\operatorname{negl}}

\newcommand{\Tr}{\operatorname{Tr}}

\newcommand{\Span}{\operatorname{span}}

\newcommand{\bit}{\{0,1\}}
\newcommand{\proj}[1]{\ket{#1}\!\bra{#1}}
\renewcommand{\vec}[1]{\mathbf{#1}}

\newcommand{\ket}[1]{\lvert #1\rangle}
\newcommand{\bra}[1]{\langle #1\rvert}

\newcommand{\ketbra}[2]{\lvert #1\rangle\!\langle #2\rvert}
\newcommand{\norm}[1]{\left\lVert #1\right\rVert}
\newcommand{\abs}[1]{\left\lvert #1\right\rvert}

\newcommand{\HaarU}{\mu_H}

\newcommand{\da}[1]{\delta_{#1}}

\newcommand{\Binf}{\beta_\infty}

\newif\ifnotes
\notesfalse      

\title{On the pseudorandomness of simple quantum processes}

\author{Jesko Dujmovic\footnote{Boston University. \href{mailto:mail@ind-jesko.net}{\texttt{mail@ind-jesko.net}}} \and
Jonas Haferkamp\footnote{Ruhr-University Bochum. \href{mailto:jonas.haferkamp@rub.de}{\texttt{jonas.haferkamp@rub.de}}}
\and Alexander Poremba\footnote{Boston University. \href{mailto:poremba@bu.edu}{\texttt{poremba@bu.edu}}}}
\date{}

\begin{document}

\maketitle

\begin{abstract}
Can simple processes appear highly complex?
Gowers (\emph{Comb. Prob. Comp. '96}) conjectured that repeatedly composing local random reversible operations can yield global permutations that are indistinguishable from random.
In this work, we study the unitary \emph{quantum} analog of this question, in an attempt to make new progress on this longstanding conjecture.

Our first result shows that statistical moment matching in the form of unitary designs does not generically lead to pseudorandomness---even for the simplest quantum processes: for every fixed $t$, we give an efficiently samplable family $\{\nu_n\}_n$ of distributions on one- and two-qubit gates such that, after $T=O_t(n^2\log^2 n)$ independent steps, the resulting $n$-qubit ensemble is an approximate unitary $t$-design with negligible error $\exp(-\Omega(\log^2 n))$, yet an efficient quantum algorithm distinguishes it from random using only $O_t(\log^2 n)$ queries.
This refutes the \emph{unitary} analog of the Hoory--Magen--Myers--Rackoff conjecture (\emph{ICALP '04}) for permutations.

Our second result is a stronger separation between unitary designs and pseudorandom unitaries at polynomially bounded moments; our counterexample, however, requires highly structured ensembles, in contrast with the simple \emph{local} walks from before. This suggests caution when using unitary designs to model information scrambling in black-hole physics, as even \emph{maximally scrambled} systems can exhibit structure which is accessible to efficient experiments.
Motivated by these findings, we then propose new conjectures for how pseudorandomness can plausibly emerge
within simple quantum processes, such as random quantum circuits.
\end{abstract}

\section{Introduction}

The seminal work of Gowers~\cite{Gowers_1996,gowers2013how} shows that random reversible circuits can achieve remarkable mixing properties: repeatedly composing random three-bit scramblers can produce global permutations on $n$-bit strings that match statistical properties of truly random permutations  (see \Cref{fig:random-circuit}). Specifically, they become
approximately $t$-wise independent after $\widetilde O\!\left(n^3t(n^2+t)(n^3+t)\right)$ rounds: for any $t$ distinct inputs, the joint distribution of their
images is close to uniform over ordered tuples of $t$ distinct outputs. These findings have had many applications across many areas of computer science; for example, informing the design of block ciphers~\cite{noauthororeditor}, cryptographic protocols~\cite{canetti2025towards,canetti2026encrypt}, hash functions~\cite{BroderCharikarFriezeMitzenmacher2000}, sampling-without-replacement methods~\cite{Thorup2013BottomK} and derandomization techniques~\cite{mohanty2020explicit}.

Gowers conjectured~\cite{gowers2013how} that these mixing properties suggest an even stronger possibility: a sufficiently deep random reversible
circuit may actually be indistinguishable from a uniformly random permutation to \emph{every} efficient observer.
In other words, they could form ensembles of cryptographic \emph{pseudorandom permutations}~\cite{LubyRackoff1988}.
 This is known as the \emph{Gowers conjecture}.


While approximate $t$-wise independence only constrains
the joint outputs on a prescribed number of inputs in a statistical sense, pseudorandomness requires security against \emph{all} polynomial-time distinguishers,
each of which may make its own polynomial number of queries. Importantly, passing statistical tests that show limited statistical independence, by itself, does not establish this stronger guarantee.
Despite substantial recent progress on the statistical mixing of random circuits~\cite{Brand_o_2016,chen2025incompressibility,gay2025pseudorandomnesspropertiesrandomreversible,baer2026random}, the computational
conjecture still remains wide open as a positive resolution would imply the existence of one-way functions, and hence $\mathsf{P}\neq \mathsf{NP}$.

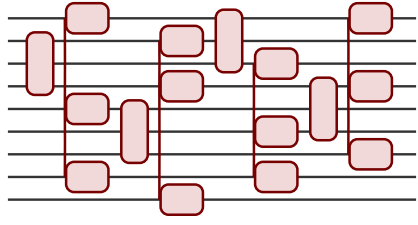
\begin{figure}[t]
    \centering
    \begin{tikzpicture}[
        scale=0.4,
        x=1.25cm,
        y=0.75cm,
        wire/.style={
            line width=0.9pt,
            draw=black!80
        },
        scrambler/.style={
            draw=DarkRed!70!black,
            fill=DarkRed!15,
            rounded corners=3pt,
            line width=0.95pt
        },
        nonlocalbox/.style={
            scrambler,
            minimum width=0.56cm,
            minimum height=0.40cm,
            inner sep=0pt
        },
        nonlocallink/.style={
            draw=DarkRed!70!black,
            line width=0.95pt
        }
    ]

    \def\circuitwidth{10.8}
    \def\localboxwidth{0.70}

    \foreach \i in {1,...,9} {
        \draw[wire] (0,-\i) -- (\circuitwidth,-\i);
    }

    %
    \newcommand{\localgate}[2]{%
        \draw[scrambler]
            (#1-\localboxwidth/2,-#2+0.38)
            rectangle
            (#1+\localboxwidth/2,-#2-2.38);
    }

    %
    \newcommand{\nonlocalgate}[5]{%
        \node[nonlocalbox] (#1a) at (#2,-#3) {};
        \node[nonlocalbox] (#1b) at (#2,-#4) {};
        \node[nonlocalbox] (#1c) at (#2,-#5) {};

        \draw[nonlocallink]
            (#1a.west)
            .. controls ({#2-0.59},-{#3+0.8})
                       and ({#2-0.59},-{#4-0.8}) ..
            (#1b.west)
            .. controls ({#2-0.59},-{#4+0.8})
                       and ({#2-0.59},-{#5-0.8}) ..
            (#1c.west);
    }


    \localgate{0.85}{2}

    \nonlocalgate{g2}{2.10}{1}{5}{8}

    \localgate{3.35}{5}

    \nonlocalgate{g4}{4.60}{2}{4}{9}

    \localgate{5.85}{1}

    \nonlocalgate{g6}{7.10}{3}{6}{8}

    \localgate{8.35}{4}

    \nonlocalgate{g8}{9.60}{1}{4}{7}

    \end{tikzpicture}

    \caption{A classical random reversible circuit consisting of simple
 $3$-bit scrambling operations. 
    }
    \label{fig:random-circuit}
\end{figure}

Hoory, Magen, Myers, and Rackoff~\cite{HMMR2005} conjectured that approximate \emph{4-wise} independence should already suffice for pseudorandomnness provided one uses ``locally composed''
permutations. Their rationale was that $4$-wise independence is the first threshold at which no obvious distinguishers are known. A common formulation of the conjecture assumes \emph{round locality}: each output bit of a round depends on only
constantly many input bits, although the round may act on all wires. Thus, the conjecture may be stated as follows:

\begin{conjecture}[HMMR Conjecture]
Let $\{\nu_n\}_n$ be an efficiently samplable distribution on reversible rounds on $n$ bits, each of whose
output bits depends on at most a constant number of input bits. If the composition of $T=\poly(n)$
independent samples from $\{\nu_n\}_n$ is approximately $4$-wise independent with exponentially small
error, then the resulting ensemble is a pseudorandom permutation ensemble.
\end{conjecture}

The fact that 3-wise independence is insufficient follows from algebraic examples~\cite{HMMR2005}.
The motivation behind the 4th-order threshold comes from group theory. Any set of permutations that satisfies 4-transitivity will, after enough compositions, at least result in the alternating group \cite{dixion1996permutation,cameron1999permutation}.
This made the HMMR conjecture highly plausible.
In a recent breakthrough, however, Dujmovic, Pelecanos, and Tessaro~\cite{cryptoeprint:2025/2282} refuted this round-local
formulation, even with any constant order of independence.

\paragraph{Quantum pseudorandomness.}
The classical questions above have direct quantum analogues: can simple \emph{quantum} processes, likewise, appear highly complex? Investigating these questions in the more general quantum case could potentially shed new light on the classical setting. Random quantum circuits are among the simplest models of complex quantum dynamics (see \Cref{fig:random-quantum-circuit}) and have a wide range of applications across quantum information science, such as in learning theory~\cite{Aharonov_2023,schuster2025random}, many-body experiments~\cite{schuster2025hardness,mao2025randomunitariesconserveenergy}, and especially in black-hole physics and quantum gravity~\cite{Hayden_Preskill_2007,Lashkari_2013,Roberts_2017,Liu_2018,Kim_2023,poremba_et_al:LIPIcs.ITCS.2026.109,schuster2025strong}. 

One may naturally ask whether the same conjectures hold if we replace reversible gates by local \emph{quantum} gates
and uniform permutations by \emph{unitaries} drawn from the \emph{Haar} measure---the uniform probability measure over the unitary group~\cite{Haar1933,Mele2024introductiontohaar}. In particular, one may ask:

\begin{center}
\emph{Are sufficiently deep random quantum circuits pseudorandom?}
\end{center}

The statistical counterpart of $t$-wise independence in the quantum setting is what is known as a \emph{unitary $t$-design}~\cite{dankert2009exact,gross2007evenly,cleve2015near}---an ensemble
matching the first $t$ unitary moments of the Haar measure.
A pseudorandom unitary~\cite{10.1007/978-3-319-96878-0_5,DBLP:journals/corr/abs-2404-12647,ma2025construct}, by contrast, only needs to fool observers with limited computational resources.
In the statistical setting, it is well known that local random circuits form approximate
designs in depth polynomial in $n$ and $t$~\cite{Brand_o_2016, DBLP:journals/quantum/Haferkamp22,chen2024incompressibilityspectralgapsrandom,baer2026random}. 
As in the classical case, however, the depth depends on $t$, whereas pseudorandomness requires
a \emph{single} polynomial-depth ensemble which remains secure against \emph{all} efficient quantum distinguishers. 
This distinction is precisely the driving motivation for the \emph{quantum Gowers conjecture}\footnote{To the best of our knowledge, explicit formulations of the conjecture first appeared in~\cite{bostanci2024efficientquantumpseudorandomnesshamiltonian,fefferman2025hardnesslearningquantumcircuits}.} which we hinted at earlier; it asks whether random
quantum circuits (see Figure~\ref{fig:random-quantum-circuit}) can form \emph{pseudorandom unitaries} after some sufficient polynomial depth.

Recent developments make this question especially timely. Advances in quantum complexity theory~\cite{kretschmer:LIPIcs.TQC.2021.2} provide strong evidence that pseudorandom unitaries can exist without one-way functions---the de-facto minimal assumption in theoretical cryptography~\cite{noauthororeditor}. This striking observation has led to the pursuit of quantum cryptography from new assumptions which may plausibly hold even if 
one-way function exist do not exist~\cite{khurana2024foundingquantumcryptographyquantum,bostanci2024efficientquantumpseudorandomnesshamiltonian,hiroka2024computationalcomplexitylearningefficiently,Hiroka_2025,fefferman_et_al:LIPIcs.ITCS.2026.56,carrasco2026instantiatingmicrocryptobstaclesopportunities,gulati2026buildmicrocrypt}.
The quantum Gowers conjecture is central among these assumptions, as it is arguably the simplest and most natural way to generate quantum pseudorandomness; in particular, it may plausibly hold even in a world in which $\mathsf{BQP} =\mathsf{QMA}$~\cite{kretschmer:LIPIcs.TQC.2021.2}.
At the same time, the conjecture is also widely believed to be true among quantum information theorists. Despite substantial recent progress on constructing \emph{pseudorandom unitaries} from standard cryptographic assumptions~\cite{10.1007/978-3-319-96878-0_5,cryptoeprint:2024/1470,ananth2024pseudorandom,DBLP:journals/corr/abs-2404-12647,Chen_2024,ma2025construct,schuster2025random,foxman2026quantumlazysamplingpath}, the quantum analog of the Gowers conjecture still remains wide open---just like its classical counterpart.


\begin{figure}[t]
    \centering
    \begin{tikzpicture}[
        scale=0.45,
        x=1.25cm,
        y=0.75cm,
        wire/.style={line width=0.9pt,draw=black!80},
        scrambler/.style={
            draw=DarkRed!70!black,
            fill=DarkRed!15,
            rounded corners=3pt,
            line width=0.95pt
        }
    ]
    \def\circuitwidth{10.8}
    \def\quantumboxwidth{0.70}

    \foreach \i in {1,...,9} {
        \draw[wire] (0,-\i) -- (\circuitwidth,-\i);
    }
    \newcommand{\quantumgate}[2]{%
        \draw[scrambler]
            (#1-\quantumboxwidth/2,-#2+0.32)
            rectangle
            (#1+\quantumboxwidth/2,-#2-1.32);
    }
    \foreach \x in {0.85,3.35,5.85,8.35} {
        \foreach \i in {1,3,5,7} {\quantumgate{\x}{\i}}
    }
    \foreach \x in {2.10,4.60,7.10,9.60} {
        \foreach \i in {2,4,6,8} {\quantumgate{\x}{\i}}
    }
    \end{tikzpicture}
    \caption{A simple nearest-neighbor random quantum circuit, where each box is a two-qubit quantum gate
    sampled independently from the Haar measure over the unitary group $\mathrm{U}(4)$.}
    \label{fig:random-quantum-circuit}
\end{figure}
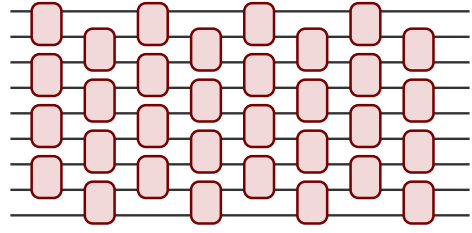

\paragraph{The Quantum HMMR conjecture.}
To what extent do random quantum circuits mix just like random classical circuits?
Interestingly, the quantum setting has a strikingly similar threshold behavior as in the classical HMMR setting~\cite{HMMR2005}. The Clifford group~\cite{nielsen2001quantum} forms an exact unitary
$3$-design~\cite{webb2016cliffordgroupformsunitary}, yet is efficiently distinguishable from a genuine Haar random unitary:
its action on a computational-basis state produces a stabilizer state, which can be recognized
efficiently~\cite{montanaro2017learningstabilizerstatesbell,GIKL2024}. 
Moreover, classification results for unitary $t$-groups show that $t=4$ is an inherent algebraic threshold at which
more complex behavior must emerge: in dimension $d \geq 3$, no finite subgroup of the unitary group $\mathrm{U}(d)$ forms an exact unitary $t$-design for $t \geq 4$~\cite{bannai2020unitary}.
The absence of a group design for the fourth moment raises an important question which is implicit in many prior works in the literature~\cite{zhu2016clifford,CHAMON2022169086,Homeopathy2023,Kaposi_2026,hunterjones2026idealrandomquantumcircuits}; namely
\begin{center}
\emph{Does Haar-like behavior of unitary designs already appear at the fourth moment?}
\end{center}
This motivates the following natural \emph{quantum analog} of the HMMR conjecture~\cite{HMMR2005}.

\begin{conjecture}[Quantum HMMR Conjecture]\label{conj-HMMR}
Let $\{\nu_n\}_n$ be an efficiently samplable family of distributions on local quantum gates, each acting on a constant number of qubits independent of $n$. Let $\nu_n^{*T}$ denote the $n$-qubit ensemble obtained by independently composing $T=\poly(n)$ gates drawn from $\nu_n$.  
If $\{\nu_n^{*T}\}_n$ forms an approximate unitary $4$-design with negligible error, then it is a pseudorandom unitary ensemble.    
\end{conjecture}

Just like its classical counterpart, the quantum HMMR conjecture asks whether limited statistical properties can be viewed as evidence of computational
pseudorandomness. The locality of the ensemble, likewise, serves as an important purpose as in the classical case: it ensures some amount of ``simplicity'' in contrast to an artificially engineered ensemble which may fail due to intricate algebraic structure~\cite{HMMR2005}. Given that the HMMR conjecture has been refuted in the classical setting, one may ask to what extent random quantum circuits admit similar counterexamples, or whether there are differences between classical and quantum that are yet to be discovered. 


\subsection{Our contributions}\label{sec:our-contributions}

Our limited progress on these conjectures reflects many important gaps in our understanding of how pseudorandomness arises within simple quantum processes.
In this work, we develop the techniques to address them; specifically, we investigate when simple, local ensembles fail to become pseudorandom despite achieving strong statistical moment properties.
This brings the study of pseudorandomness and statistical mixing of quantum processes on roughly the same theoretical footing as those of random classical circuits. 

Building on these no-go results, we then ask when genuine pseudorandomness can emerge; we explore how the physics notions of thermalization, scrambling and entanglement saturation can serve as new proxies for the emergence of pseudorandomness.
We hope that this connection can inform the Gowers conjecture in both the classical and the quantum setting.

We now give an overview of our main results. In what follows, we assume basic familiarity with standard concepts from the literature on unitary $t$-designs and Haar integration; we refer to the preliminary section in \Cref{sec:prelim} for additional background on these topics.

\paragraph{Refuting the Quantum HMMR conjecture.}

Our first contribution focuses on the natural quantum analog of HMMR conjecture~\cite{HMMR2005} which deals with \emph{local} quantum processes. We refute the Quantum HMMR conjecture, as described earlier, in full generality---essentially to the same extent as in the classical setting by~\cite{cryptoeprint:2025/2282}. This is captured by the following theorem.

\begin{theorem}[Refutation of Quantum HMMR]
\label{thm:intro-walk}\label{thm:walk}
For every fixed integer $t\geq4$, there exists an efficiently samplable family $\{\nu_n\}_n$ of distributions on one- and two-qubit gates such that, for all sufficiently large $n$, the ensemble $\nu_n^{*T}$ obtained from $T=O_t(n^2\log^2 n)$ independent steps is an approximate adaptive unitary $t$-design with negligible error $\exp(-\Omega_t(\log^2 n))$. Nevertheless, an efficient quantum algorithm distinguishes this ensemble from Haar with constant advantage using $O_t(\log^2 n)$ parallel forward queries.
\end{theorem}
This rigorously shows that statistical moment properties, in general, fail to translate into computational pseudorandomness---even for simple processes that comprise local interactions.
Interestingly, our approach deviates significantly from the classical counterexample of~\cite{cryptoeprint:2025/2282} and instead relies on inherently quantum techniques which have no classical counterpart.\footnote{Interestingly, our counterexample does not require the use of a \emph{counter} as in the classical setting~\cite{cryptoeprint:2025/2282}.}


Our approach is to apply a random nearest-neighbor Clifford gate at each step, except with some small probability $1/m$ with $m=\Theta_t(n^2)$, when we instead apply a random non-Clifford single-qubit $Z$-rotation. 
The Clifford gates provide enough mixing without destroying stabilizer structure, while the occasional rotations allow the moments to approach those of the Haar measure. This approach is inspired by the structured ensembles in~\cite{Homeopathy2023,BittelLeone26}; in our case, however, we analyze a \emph{single i.i.d.\ local walk}. A standard spectral gap analysis does not suffice for us: converting a moment bound into an adaptive-query bound incurs an exponential blow-up in $n$, and running the walk long enough to compensate for this factor would introduce too many non-Clifford gates.

The key is to use more information about the spectrum of the moment operator. We know that the local Clifford walk converges rapidly to the Clifford twirl, but this alone does not give Haar moments since the Clifford twirl fixes more operators than the Haar twirl does. However, for fixed $t$, the fixed-point space of the $t$-fold Clifford twirl has dimension bounded independently of $n$~\cite{gross2021schur}. Thus, Clifford mixing supresses most of the spectrum and only constantly many eigenvalues need to be handled by the rare non-Clifford rotations. 
In this subspace of constant dimension it then suffices to apply known spectral bounds on doped Clifford circuits~\cite{baer2026random}. 

To prove the adaptive design property, we fix an adaptive $t$-query algorithm and express the acceptance probability using the spectral decomposition of the moment operator. The exponentially small eigenvalues already compensate for the dimension factor in the usual bound. The remaining contribution is a sum of constantly many decaying exponentials, one for each of the remaining eigenvalues. Instead of bounding the coefficients separately, we exploit the fact that the first few values are bounded because they are differences of acceptance probabilities. 
This strategy is similar in spirit to the polynomial method, where the boundedness of the outcome probabilities themselves is exploited to argue about the distinguishability of two quantum algorithms.
The boundedness of the outcome probability was also directly used in~\cite{BittelLeone26}.

\paragraph{Separating unitary designs from pseudorandom unitaries.}

The above counterexample to the quantum HMMR conjecture relies crucially on the fact that the ensemble only matches \emph{constant} moments of the Haar measure. Could an ensemble that matches much higher (say, polynomial) moments of the Haar measure be sufficient to achieve pseudorandomness? Clearly, matching increasingly higher moments seems to rule out more and more statistical tests. We prove that that the answer is \emph{no} in general; however, to show this we must deviate from the prior setting of \emph{local} ensembles and explore new constructions. While this statement is well-known classically, the quantum case appears much more subtle and has remained open for quite some time.

Our second contribution is the (to the best of our knowledge) first rigorous separation between arbitrary polynomial-order unitary designs and pseudorandom unitaries.
\begin{theorem}[Separating Adaptive Unitary Designs from Pseudorandom Unitaries]
\label{thm:intro-phase}\label{thm:phase}
Let $n\geq64$ and $1\leq t\leq 2^{n/4-4}$ be an integer (depending on $n$). For every
$0<\varepsilon\leq1/2$, there exists an adaptive unitary $t$-design in depth $\poly(n,t,\log(1/\varepsilon))$ with error $\frac{2t(t-1)}{2^n}+\varepsilon$, and yet an efficient quantum algorithm using $O(nt)$ forward queries
and $\poly(n,t)$ time can distinguish it from Haar
with constant advantage.
\end{theorem}

The basic idea behind the specific circuit ensemble is to engineer a special fixed \emph{subspace} in which the ensemble retains a controlled amount of detectable structure, while exhibiting strong mixing behavior on the orthocomplement of the subspace.

We now describe our circuit ensemble in more detail. First, we pick the one-dimensional subspace $S = \Span\{ \ket{+}^{\otimes n}\}$. Our construction then uses what we call the ``$FB$ ensemble'': we construct unitaries of the form $U=FB$, where $F$ is a diagonal phase function circuit, and where $B$ is a block-diagonal circuit that leaves $S$ invariant and mixes its orthogonal complement.
Before we explain the full construction in detail, we will first consider a simpler ensemble of the form $U = F B_\infty$, for an idealized unitary $B_\infty$. This is a product of two unitary operators, where
\begin{itemize}
    \item $F=\sum_{x }(-1)^{f(x)}\proj{x}$ is a phase function operator that maps $\ket{+}^{\otimes n}$ to
a binary phase state $\ket{\psi_f} = 2^{-n/2} \sum_x (-1)^{f(x)} \ket{x}$, for a $2t$-wise independent function\footnote{Concretely, we choose $f(x)=\mathrm{tr}_{\F_{2^n}/\F_2}\!\left(\sum_{j=0}^{2t-1}a_jx^j\right)$ for uniformly random and independent coefficients $a_j\in\F_{2^n}$. As we explain in \Cref{sec:phase-construction}, this ensemble requires merely $2nt$ seed bits and is $2t$-wise independent.}  $f$, and

\item $B_\infty$ is a Haar random unitary which leaves the subspace $S$ invariant; in other words, $B_\infty$ is a block-diagonal unitary such that $B_\infty = 1_S \oplus B_{S^\perp}$, where $ B_{S^\perp}\sim\HaarU$, i.e., $B_{S^\perp}$ is sampled according to the Haar measure over the unitary group acting on $S^\perp =\Span\{ \ket{+}^{\otimes n}\}^\perp$.
\end{itemize}
We now explain why this ensemble remains distinguishable from a Haar random unitary. Suppose that we query the unitary $m$ many times on input $\ket{+}^{\otimes n}$ within $S$. This results in $m$ copies of a binary phase state $\ket{\psi_f}^{\otimes m}$. By design, these binary phase states are chosen precisely so that they match strong statistical properties of the Haar measure; for example, when $m=t$, we know that $2t$-wise independence ensures that these states form a state $t$-design with error $O(t^2/2^n)$, and thus approximately match $t$ moments of the Haar measure~\cite{brakerski2019pseudo}.
However, for a slightly larger number of copies, say $m=O(nt)$, we can exploit the fact that binary phase states become \emph{learnable}, as shown by Arunachalam, Bravyi, Dutt, and Yoder~\cite{arunachalam2023optimalalgorithmslearningquantum}. Using their derivative-measurement technique, we find that each copy reveals a relative phase $f(x)+f(y)$ for a uniformly random pair of distinct inputs, which is a linear equation in the $2nt$ seed bits. A simple finite-field character-sum bound shows that only $m=O(nt)$ such equations determine the phase function up to an additive constant with high probability. Gaussian elimination then allows us to recover a candidate, and testing it on fresh copies distinguishes our ensemble from Haar with constant advantage.

Next, we explain why the invariance of $U = F B_\infty$ on the special subspace $S=\Span\{ \ket{+}^{\otimes n}\}$ does not spoil the overall unitary $t$-design property. Our key insight is that $U = F B_\infty$ can be viewed as a special case of the  "PFC ensemble" of Metger, Poremba, Sinha and Yuen~\cite{10756150}. 

To see why this is the case, notice that every computational basis permutation $P$ fixes $\ket{+}^{\otimes n}$, and hence Haar invariance on $S^\perp$ tells us that $PB_\infty$ has the same distribution as $B_\infty$. We can therefore insert an independent uniform permutation without changing the ensemble:
\[
    F B_\infty
   \quad \mapsto \quad F PB_\infty
    =P(P^\dagger F P)B_\infty.
\]
For each fixed $P$, the conjugated phases $P^\dagger F P$ remain $2t$-wise independent. Moreover, $B_\infty$ supplies the other crucial role of the Clifford \emph{collision-breaking operator}: on $t$ copies, it places all but $O(t^2/2^n)$ expected weight in the \emph{distinct subspace}, spanned by computational basis tuples with pairwise distinct labels, uniformly over the input. On this subspace, phase averaging removes all matrix entries except those relating permutations of the same tuple, while the random basis permutation makes these tuples uniform. To formally carry out the proof of the approximate adaptive unitary $t$-design property, we rely on a recent
positive-operator analysis of the PFC construction due to Cui, Schuster, Brand\~ao, and Huang~\cite{CuiEtAl2025}. 
Their proof technique allows us to turn this collision bound into an $O(t^2/2^n)$ unitary $t$-design error---even with respect to \emph{adaptive} $t$-query experiments, without a square-root loss which is inherent in~\cite{DBLP:journals/corr/abs-2404-12647}. This completes our proof overview of the idealized $U = F B_\infty$ ensemble, and shows that it achieves the statistical moment property, despite its efficiently detectable structure on the special subspace $S=\Span\{ \ket{+}^{\otimes n}\}$.

Finally, we replace the idealized operator $B_\infty$ by an explicit random circuit that fixes $\ket{+}^{\otimes n}$. The key step is a decomposition of every special unitary on $S^\perp$ into $2n-1$ operations, each applying a unitary to $n-1$ qubits conditioned on the remaining qubit being in the state $\ket{-}$. We combine this decomposition with controlled versions of the constant-depth generators of~\cite{anshu2026depth}. Next, we use bounds on Kazhdan constants which result in a walk on $O(n)$ efficiently implementable circuits with an inverse-polynomial spectral gap. This gap lets us approximate the $t$ moments of $B_\infty$ in depth $\poly(n,t,\log(1/\varepsilon))$, with error small enough to control adaptive $t$-query experiments. Since every step fixes $\ket{+}^{\otimes n}$, the phase-learning distinguisher still applies to $U=FB$. Unlike our first construction, however, the walk samples \emph{controlled} circuits rather than individual local gates, and hence the separate phase circuit $F$ remains essential for the construction.

\paragraph{Pseudorandomness and maximal scrambling.}

Our two no-go results suggest that both locality and the number of matched moments play a role in the emergence of pseudorandomness:
\begin{itemize}
    \item \emph{locality} constrains the kinds of algebraic structure that can arise, and in turn limits what an efficient distinguisher can exploit, as in the $FB$ ensemble; the more structure we build into the ensemble, the more opportunities we give a distinguisher to take advantage of it. 
    
    \item \emph{moment matching} is equally important to defeat broad classes of statistical tests, as repeated local dynamics do not provide a guarantee on randomness. Our quantum HMMR counterexample shows that local iterated ensembles can still retain efficiently detectable stabilizer structure when it matches only a constant number of moments of the Haar measure.
\end{itemize}
We should therefore try to match more moments while keeping the process local. But how many moments should we match before we can reasonably expect pseudorandomness to kick in? 

In \Cref{sec:scrambling}, we explore an explicit connection to the notion of \emph{scrambling} from many-body physics. Scrambling describes how initially local information spreads across many degrees of freedom until it becomes inaccessible to local observations. This is closely related to \emph{diffusion} in classical block ciphers, where changes in a few input bits spread across many output bits~\cite{HMMR2005,noauthororeditor}. The connection to unitary designs provides a natural scale for our question: Liu, Lloyd, Zhu, and Zhu~\cite{Liu_2018} showed that higher moment matching leads to saturation at progressively higher orders in a hierarchy of entanglement entropies. Eventually, this process reaches a limit: when $t= \Theta(n)$ all entanglement entropies are saturated, which is known as \emph{maximal scrambling}.

Our results suggest, however, that even maximal scrambling does not imply pseudorandomness, at least \emph{generically}. The $FB$ ensemble can reach this regime, with expected min-entropy within a constant of the maximum across every fixed balanced crossed cut of its Choi state, yet remain efficiently distinguishable from Haar with constant advantage using only $O(n^2)$ forward queries. Increasing the depth of $B$ improves mixing without changing the relation $U\ket{+}^{\otimes n}=F\ket{+}^{\otimes n}$, so the same phase-learning distinguisher continues to work. This suggests caution when using unitary designs to model scrambling in quantum many-body physics and black-hole evolution: even \emph{maximally scrambled} systems may exhibit some amount of structure that is accessible to efficient experiments. This challenges the notion that maximal scrambling is synonymous with Haar behavior.  Pseudorandomness, on the other hand, still remains a reliable proxy. Can maximal scrambling still teach us something about pseudorandomness?

Neither counterexample rules out pseudorandomness when independent \emph{local} dynamics match moments up to the maximal scrambling scale, which suggests that a genuine threshold behavior may occur. We therefore conjecture that an ensemble obtained by independently composing $T=\poly(n)$ local gates drawn from the same efficiently samplable distribution is pseudorandom whenever it forms an approximate unitary $t$-design for $t=\Theta(n)$ with negligible additive error. 

Known results on local random quantum circuits already suffice to achieve this regime, giving polynomial-depth designs at $t=\Theta(n)$~\cite{Brand_o_2016}, so this conjecture offers a possible route towards the quantum Gowers conjecture. Recent results showing uniformly small outcome probabilities at the same design order provide further evidence for the significance of this threshold~\cite{hunterjones2026idealrandomquantumcircuits}. More generally, this connection suggests studying growing-order independence in local reversible circuits as a way to move beyond the classical HMMR counterexamples, and may provide a new perspective on the Gowers conjecture in both the classical and the quantum setting.

\section{Related work}
\label{sec:related-work}

The closest classical precursor to our work is that of Dujmovic, Pelecanos, and Tessaro~\cite{cryptoeprint:2025/2282}, which refuted the conjecture of Hoory, Magen, Myers, and Rackoff~\cite{HMMR2005} by constructing locally composed permutations with approximate independence of any prescribed constant order that fail to become pseudorandom. Their notion of locality allows a round to act on all wires, provided each output bit depends on only constantly many input bits. Our first result gives a quantum analogue under the requirement that each step acts on at most two qubits. Our proof, however, uses inherently quantum techniques to balance moment matching and stabilizer structure.

\cite{Homeopathy2023} shows that unitary designs can be obtained with a number of non-Clifford gates independent of the system size, for fixed design order and error. Bittel and Leone~\cite{BittelLeone26} strengthen this to adaptive queries and give an efficient test for the stabilizer structure that survives at low non-Clifford count. These results provide distinguishable designs beyond the Clifford $3$-design example~\cite{webb2016cliffordgroupformsunitary,GIKL2024}. Our construction deviates from these two works as it obtains a negligible error adaptive design via a homogenous i.i.d.\ walk on only one- and two-qubit gates. 

Our higher-order separation is related to constructions and learning algorithms for quantum phase states. Brakerski and Shmueli~\cite{brakerski2019pseudo} showed that binary phase states defined by $2t$-wise independent functions form approximate state $t$-designs. Arunachalam, Bravyi, Dutt, and Yoder~\cite{arunachalam2023optimalalgorithmslearningquantum} developed derivative measurements for learning polynomial phase states, which we use to recover the seed of our phase family. Our second main result extends these learning algorithms into a fine-grained query setting. The analysis of our second ensemble crucially builds on the PFC construction of Metger, Poremba, Sinha, and Yuen~\cite{10756150} and the adaptive-query methods of Cui et al.~\cite{CuiEtAl2025}. Our work shows how to embed structure into the PFC ensemble in a way that does not spoil the design property.

A different obstruction to pseudorandomness was established by Cui et al.~\cite{cui2025randomunitarieshamiltoniandynamics}, who showed that evolution under constant-local, time-independent Hamiltonians is efficiently distinguishable from Haar measure, irrespective of the evolution time. Their obstruction already prevents sufficiently accurate unitary $2$-designs. Importantly, it does not separate moment matching from pseudorandomness: both properties fail in their setting. Our local-circuit counterexample instead achieves any prescribed constant design order with negligible additive error, so its distinguishability cannot be explained by such a low-moment obstruction.

Liu, Lloyd, Zhu, and Zhu~\cite{Liu_2018} showed that statistical moments can be viewed as a proxy for Haar randomness through the lens of entanglement entropies: they showed that increasing design order controls progressively higher R\'enyi entanglement entropies, with sufficiently accurate designs of order $\Theta(n)$ giving near-maximal expected min-entropy. In our work, we  discuss explicit connections to computational pseudorandomness.

Concurrent work of Raza, Eisert, and Fefferman~\cite{raza2026distinctness} studies the relationship between statistical and computational pseudorandomness from a complementary direction: they show that design-level randomness is not necessary for constructing pseudorandom unitaries, identifying a substantially weaker ``distinctness'' property that suffices in the non-adaptive setting. Our results show that even strong unitary-design guarantees need not be sufficient for pseudorandomness, so together the two works delineate the gap between statistical and computational notions of quantum randomness from opposite sides.

\paragraph{AI Statement.}

The conceptual ideas behind all of the main results are due to the authors. Prior to the release of ChatGPT Astra, we had obtained preliminary versions of \Cref{thm:intro-walk} and \Cref{thm:phase} without substantial use of AI tools. After its release, we used ChatGPT Astra to improve and simplify our constructions further, so as to obtain the best bounds and parameters for all of our theorems. The model was used to assist with some parts of the technical write-up, for devising figures and for the literature review.  We give special credit to ChatGPT Astra for its contribution to the formulation and proof of \Cref{lem:walk-scalar-decay}. The authors checked the results and take full responsibility for their correctness.

\paragraph{Acknowledgements.}

We thank Ran Canetti and Angelos Pelecanos for many useful discussions.

\paragraph{Organization.} In \cref{sec:prelim}, we present mathematical preliminaries and a minimal background on quantum information.
The proof behind the counterexample to the quantum HMMR conjecture in \Cref{thm:intro-walk} appears in \Cref{sec:walk-proof}. Our separation between unitary designs and pseudorandom unitaries in \Cref{thm:phase} is shown in~\Cref{sec:phase}. In \Cref{sec:scrambling}, we explore connections to many-body physics and scrambling, and present new conjectures for when pseudorandomness can emerge.

\section{Preliminaries}
\label{sec:prelim}

 We write $[n]=\{1,\ldots,n\}$ and
$|S|$ for the cardinality of a finite set $S$. If $X$ is distributed
according to a probability measure $\mu$, we write $X\sim\mu$ and
$\E_{X\sim\mu}$ for the corresponding expectation. We use
$O_t(\cdot)$, $\Omega_t(\cdot)$, and $\Theta_t(\cdot)$ when the implicit
constants may depend on $t$. A function $\negl:\mathbb{N}\to\mathbb{R}_{\ge 0}$ is \emph{negligible} if, for every polynomial $p$, we have $\negl(n)<\frac{1}{p(n)}$,
for all sufficiently large $n$.

\paragraph{Operators, norms, and vectorization.}

We use $\mathrm{L}(V)$ to denote the set of linear operators acting on a finite-dimensional complex vector spaces $V$. The $d$-dimensional identity operator is $I_d$.
For $X \in \mathrm{L}(\mathbb{C}^d)$, we use $\|X\|_1=\Tr\sqrt{X^\dagger X}$ for the trace norm,
$\|X\|_2=(\Tr X^\dagger X)^{1/2}$ for the Frobenius norm,
$\|X\|_\infty$ for the operator norm.

\paragraph{Haar measure and moment operators.}
The Haar measure $\HaarU$ is the unique (left and right) invariant uniform probability measure over the unitary group $\mathrm{U}(d)$; for any integrable function $f$, it comes with the invariance~\cite{Haar1933,simon1996representations} of the form
\begin{align}
\int_{\mathrm{U}(d)} f(UV) \mathrm{d}\,\HaarU(U)  = \int_{\mathrm{U}(d)} f(U) \mathrm{d}\,\HaarU(U)   = \int_{\mathrm{U}(d)} f(VU) \mathrm{d}\,\HaarU(U).
\end{align}
For brevity, we often write $\int_{\mathrm{U}(d)} f(U) \mathrm{d}\,\HaarU(U) = \underset{U \sim \HaarU}{\E}[f(U)]$, where we use the notation $U \sim \HaarU$ to describe that the unitary $U$ is sampled uniformly according to the Haar measure over $\mathrm{U}(d)$. More generally, for a unitary ensemble $\nu$, the notation $U\sim\nu$ means that $U$ is sampled according to $\nu$.
For $U\sim\nu$ and $V\sim\eta$ sampled independently, $\nu*\eta$
is the distribution of $UV$.
The notation $\nu^{*r}$ describes a product of $r$ independent samples $(U_r U_{r-1} \dots U_1)$, where $U_i \sim \nu$ for $i\in[r]$.

Given a unitary ensemble $\nu$ and $t \in \mathbb{N}$, we define the $t$-fold twirl $\mathcal{T}_{\nu}^{(t)}$ over $\nu$ as the map
\begin{equation}
 X \quad \mapsto \quad \mathcal{T}_{\nu}^{(t)}(X)=\underset{U \sim \nu} {\E}\left[U^{\otimes t}X(U^\dagger)^{\otimes t}\right].
 \label{eq:twirl}
\end{equation}
We denote the corresponding vectorized moment operator of order $t$ as
\begin{equation}
\mathcal{M}_{\nu}^{(t)}=\underset{U \sim \nu} {\E} \left[
 U^{\otimes t}\otimes\overline U^{\otimes t} \right].
 \label{eq:moment}
\end{equation}
Thus $\mathrm{vec}(\mathcal T_\nu^{(t)}(X))
=\mathcal M_\nu^{(t)}\mathrm{vec}(X)$ for every $X$.
We use $\mathcal{T}_{\mu_H}^{(t)}$ and $\mathcal{M}_{\mu_H}^{(t)}$ to denote the $t$-fold twirl and moment operator for the Haar measure. 

\paragraph{Adaptive unitary designs.}

A $t$-query quantum oracle-algorithm $\mathcal{A}^U$ with access to an $n$-qubit unitary oracle $U \in \mathrm{U}(2^n)$ is a sequence of $(n+m)$-qubit unitaries of the form
\begin{align}
A_t \cdot U \cdot A_{t-1} \cdot U \cdot A_{t-2} \cdots A_1 \cdot U \cdot A_0  
\end{align}
where $A_t,A_{t-1},\dots,A_1,A_0$ act on registers $AB$, for an $n$-qubit register $A$ and $m$-qubit ancilla register $B$, for $m \in \mathbb{N}$. The \emph{state} of a $t$-query oracle algorithm $\mathcal{A}^U$ on input $\ket{0^{n+m}}$ is defined as
\begin{align}
\ket{\mathcal{A}_t^U} = A_t \cdot U \cdot A_{t-1} \cdot U \cdot A_{t-2} \cdots A_1 \cdot U \cdot A_0  \ket{0^{n+m}}. 
\end{align}

An (adaptive) approximate unitary $t$-design approximately matches the first $t$ moments of the Haar measure. We give a formal definition below.




\begin{definition}[Approximate unitary $t$-design]
Let $n, t \in \mathbb{N}$. An ensemble of $n$-qubit unitaries $\nu$ is an $\epsilon$-approximate adaptive unitary $t$-design if any adaptive $t$-query quantum oracle algorithm $\mathcal{A}^U$ can distinguish $U \sim \nu$ from a Haar random unitary $U \sim \HaarU$ with advantage at most $\epsilon$; taking the supremum over all such $\mathcal{A}^U$ with two-outcome measurements $(M_{\mathcal{A}},I-M_{\mathcal{A}})$, this means
\begin{align}
\da{t}(\nu,\mu_H):=\sup_{\mathcal A,M_{\mathcal A}}\left|\underset{U \sim \nu} {\E} \Tr\left[ M_{\mathcal{A}} \cdot \proj{\mathcal{A}_t^U}\right]  - \underset{U \sim \HaarU} {\E} \Tr\left[ M_{\mathcal{A}} \cdot\proj{\mathcal{A}_t^U}\right] \right| \leq \epsilon.
\end{align}
\end{definition}

Next, we introduce the computational analog of \emph{pseudorandom unitaries}~\cite{10.1007/978-3-319-96878-0_5,DBLP:journals/corr/abs-2404-12647,ma2025construct}. This requires us to formally define the notion of a quantum polynomial-time algorithm.

\paragraph{Pseudorandom unitaries.} 

A quantum polynomial-time (QPT) algorithm $\mathcal{A}$ is a uniform family of circuits $\{\mathcal{A}_n\}_{n \in \mathbb{N}}$ such that there is a polynomial $p(n)$ with the property that the circuit size of $\mathcal{A}_n$ is bounded by $p(n)$; here, uniform means that there exists a polynomial-time classical algorithm which, on input $1^n$, outputs a classical description of $\mathcal{A}_n$. 
A QPT algorithm, may also initialize some number of
ancilla qubits to $\ket{0}$ and trace them out at the end of the computation.

\begin{definition}[Pseudorandom unitary]\label{def:PRU} A unitary ensemble of $\{\mathcal{U}_n\}_{n \in \mathbb{N}}$ is called a pseudorandom unitary if, for all $n$, $\mathcal{U}_n = \{U_k\}_{k \in \mathcal{K}_n}$, is a set of $n$-qubit unitaries with key space $\mathcal{K}_n$ such that
\begin{itemize}
    \item (Efficient generation:) There exists a QPT algorithm running in time $\poly(n)$ that implements the $n$-qubit unitary $U_k$ on input $k$, for any key $k \in \mathcal{K}_n$.

    \item (Pseudorandomness:) For any QPT oracle-algorithm $\{\mathcal{A}_n^{(\cdot)}\}_{n \in \mathbb{N}}$ which runs in time $\poly(n)$ and makes $\poly(n)$ queries, followed by an efficient two-outcome measurement $(M_{\mathcal{A}_n},I-M_{\mathcal{A}_n})$,
$$
\left|\underset{k \sim \mathcal{K}_n} {\E} \Tr\left[ M_{\mathcal{A}_n} \cdot \proj{\mathcal{A}_n^{U_k}}\right]  - \underset{U \sim \HaarU} {\E} \Tr\left[ M_{\mathcal{A}_n} \cdot\proj{\mathcal{A}_n^U}\right] \right| \leq \negl(n).
$$
    
\end{itemize}
\end{definition}

\paragraph{Distinct subspace.}
The set of $t$-fold tuples over $[d]=\{1,2,\dots,d\}$ is denoted by $[d]^t$. We call a tuple $\vec x = (x_1,\dots,x_t) \in [d]^t$ distinct if $x_i \neq x_j$ for all $i \neq j$. We denote the set of distinct tuples in $[d]^t$ by $\mathrm{dist}(d,t)$.
We define the projector onto the $t$-fold distinct subspace as
\begin{equation}
 \Pi_{\mathrm{dist}(d,t)} =\sum_{\substack{x_1,\ldots,x_t\in[d]\\
                         x_i\ne x_j\ (i\ne j)}}
 \ketbra{x_1,\ldots,x_t}{x_1,\ldots,x_t}.
 \label{eq:distinct-projector}
\end{equation}
For example, the projector onto the two-fold distinct subspace is $\Pi_{\mathrm{dist}(d,2)}=I-\sum_x\ketbra{x,x}{x,x}$.

\theoremstyle{plain}
\newcommand{\maintheoremname}{}
\newtheorem*{mainresult}{\maintheoremname}
\newenvironment{maintheorem}[2]{%
  \renewcommand{\maintheoremname}{Theorem~\ref{#1}}%
  \begin{mainresult}[#2]}{\end{mainresult}}
\theoremstyle{definition}

\medskip

\section{Refuting the Quantum HMMR conjecture}
\label{sec:walk}

We prove \cref{thm:intro-walk} by adding \emph{rare} single-qubit
non-Clifford gates to a local Clifford walk. The Clifford updates produce
rapid mixing, while the non-Clifford gates handle the remaining
low-order moment obstructions. Because non-Clifford gates appear infrequently,
the resulting circuits still retain enough stabilizer structure to be
efficiently distinguished from the Haar measure.

Let $\overline{\mathcal{C}}_n$ denote the $n$-qubit Clifford group modulo
global phase, equipped with the uniform distribution whenever it is used
as an ensemble. Fix a finite inverse-closed generating set
$\mathcal G\subseteq\overline{\mathcal{C}}_2$ containing the identity.
We arrange the qubits on a ring, with indices understood modulo $n$.
For an integer $m\geq2$, one update is sampled as follows:
\begin{itemize}
 \item With probability $1-1/m$, choose an edge $(i,i+1)$ uniformly
 among the $n$ edges and apply a random gate from $\mathcal G$
 to those two qubits. We denote this local Clifford distribution by
 $\chi_n$.
 \item With probability $1/m$, apply a uniformly random rotation\footnote{Note that, for fixed $t$, a simple discretization of the rotation angles reproduces the continuous
$t$-th moment exactly, while finite-precision implementation
introduces only negligible total error. We omit these implementation details.} $e^{i\phi Z}$
to the first qubit. We denote this single-qubit
 distribution by $\xi_n$.
\end{itemize}
Thus, the one-step distribution of our walk is of the form
\begin{equation}
 \nu_n=(1-m^{-1})\chi_n+m^{-1}\xi_n.
 \label{eq:walk-step-distribution}
\end{equation}
We call a choice of the $\xi_n$ branch an \emph{insertion}, including
its identity outcome. Every update acts on at most two neighboring
qubits, and all updates are independent samples from the same
distribution.

\begin{maintheorem}{thm:intro-walk}{Refutation of Quantum HMMR}
For each fixed integer $t\geq4$, there are an integer $A_t\geq1$ and a
constant $\kappa_t>0$ such that, for sufficiently large $n$, setting
\begin{equation}
 m=A_tn^2,\qquad
 T=3m\lceil\log^2 n\rceil=O_t(n^2\log^2 n)
 \label{eq:walk-length}
\end{equation}
results in an iterated ensemble $\nu_n^{*T}$ based on the local ensemble $\nu_n$ in \eqref{eq:walk-step-distribution}, and which forms an adaptive unitary $t$-design with negligible error
\begin{equation}
 \da t(\nu_n^{*T},\HaarU)\leq e^{-\kappa_t\log^2 n}.
 \label{eq:walk-design}
\end{equation}
Nevertheless, a uniform polynomial-time algorithm distinguishes this
ensemble from Haar with constant advantage using
$O(\log^2 n)$ parallel forward queries.
\end{maintheorem}

We need the following standard lemma that writes the expected outcome probability of a quantum experiment as a linear function in the moment operator: 

\begin{lemma}[Acceptance probabilities from moment operators]
\label{lem:walk-query-functional}
Let $t \in \mathbb{N}$ and let $\mathcal A^U$ be any adaptive $t$-query algorithm $\mathcal A^U$ with forward access to
an $n$-qubit unitary, and which performs a two-outcome measurement
$(M_{\mathcal A},I-M_{\mathcal A})$. Let $d=2^n$. Then, there exists a matrix
$\Lambda_{\mathcal A}$, depending only on the algorithm $\mathcal{A}$ and its measurement,
such that, for every unitary ensemble $\nu$,
\begin{equation}
 \underset{U\sim\nu}{\E}  \Tr\bigl[M_{\mathcal A}\cdot \ketbra{\mathcal A^U_t}{\mathcal A^U_t}\bigr]=\Tr\!\left[\Lambda_{\mathcal A}^{\dagger} \cdot
                   \mathcal M_\nu^{(t)}\right],
 \qquad \|\Lambda_{\mathcal A}\|_1\leq d^{2t}.
 \label{eq:walk-query-trace}
\end{equation}
\end{lemma}
\begin{proof}
For strings $\mathbf a=(a_1,\dots,a_t),\mathbf b=(b_1,\dots,b_t)\in[d]^t$ define
\[
  \ket{\phi_{\mathbf a,\mathbf b}}
  := A_t\,(\ketbra{a_t}{b_t}\otimes I_B)\,A_{t-1}\cdots A_1\,
     (\ketbra{a_1}{b_1}\otimes I_B)\,A_0\ket{0^{n+m}} .
\]
Writing $U\otimes I_B=\sum_{a,b\in[d]}U_{ab}\ketbra{a}{b}\otimes I_B$ in each
oracle call yields
\[
  \ket{\mathcal A^U_t}=\sum_{\mathbf a,\mathbf b\in[d]^t}
  \Bigl(\prod_{i=1}^t U_{a_ib_i}\Bigr)\ket{\phi_{\mathbf a,\mathbf b}} .
\]
Let $G$ be the $d^{2t}\times d^{2t}$ matrix with entries
$G_{(\mathbf a,\mathbf b),(\mathbf a',\mathbf b')}
=\bra{\phi_{\mathbf a,\mathbf b}}M_{\mathcal A}\ket{\phi_{\mathbf a',\mathbf b'}}$. Then
\[
  \Tr\bigl[M_{\mathcal A}\ketbra{\mathcal A^U_t}{\mathcal A^U_t}\bigr]
  =\sum_{\mathbf a,\mathbf b,\mathbf a',\mathbf b'}
   \Bigl(\prod_{i}U_{a_ib_i}\overline{U_{a'_ib'_i}}\Bigr)
   G_{(\mathbf a',\mathbf b'),(\mathbf a,\mathbf b)} .
\]
Since
$\bra{\mathbf a,\mathbf a'}U^{\otimes t}\otimes\overline U^{\otimes t}\ket{\mathbf b,\mathbf b'}
=\prod_i U_{a_ib_i}\overline{U_{a'_ib'_i}}$ and $G$ is Hermitian, this equals
$\Tr[\Lambda_{\mathcal A}^\dagger\,U^{\otimes t}\otimes\overline U^{\otimes t}]$ for the
realigned matrix
\[
  (\Lambda_{\mathcal A})_{(\mathbf a,\mathbf a'),(\mathbf b,\mathbf b')}
  := G_{(\mathbf a,\mathbf b),(\mathbf a',\mathbf b')} .
\]
Averaging over $U\sim\nu$ proves formula.

For the norm bound, $G\succeq0$ is the Gram matrix of the vectors
$M_{\mathcal A}^{1/2}\ket{\phi_{\mathbf a,\mathbf b}}$, and $M_{\mathcal A}\preceq I$ gives
$\Tr G\le\sum_{\mathbf a,\mathbf b}\|\phi_{\mathbf a,\mathbf b}\|^2$.
The map $X\mapsto\sum_{a,b}(\ketbra ab\otimes I_B)X(\ketbra ba\otimes I_B)
=I_A\otimes\Tr_A X$ multiplies the trace by $d$, and the $A_i$ preserve it.
Applying this at each of the $t$ oracle positions gives
$\sum_{\mathbf a,\mathbf b}\|\phi_{\mathbf a,\mathbf b}\|^2=d^t$, so
$\|G\|_1=\Tr G\le d^t$. Realignment only permutes entries, so
$\|\Lambda_{\mathcal A}\|_2=\|G\|_2\le\|G\|_1\le d^t$. Finally,
$\Lambda_{\mathcal A}$ is a $d^{2t}\times d^{2t}$ matrix, hence
\[
  \|\Lambda_{\mathcal A}\|_1\le d^t\|\Lambda_{\mathcal A}\|_2\le d^{2t}.
\]
\end{proof}

The next lemma summarizes results which are shown in~\cite{gross2021schur} and~\cite{baer2026random}:

\begin{lemma}[Spectral gaps of the Clifford walk]
\label{lem:walk-homeopathy}
\label{lem:walk-clifford-input}
Fix $t\geq4$. The moment operators
$\mathcal M_{\overline{\mathcal{C}}_n}^{(t)}$ and
$\mathcal M_{\HaarU}^{(t)}$ are orthogonal projectors onto the Clifford-
and Haar-invariant subspaces, respectively, with the latter contained
in the former. For all sufficiently large $n$,
\begin{equation}
 \operatorname{rank}\mathcal M_{\overline{\mathcal{C}}_n}^{(t)}
 =D_t:=\prod_{j=0}^{t-2}(2^j+1).
 \label{eq:walk-clifford-rank}
\end{equation}
There are universal constants $\gamma,a>0$ such that the following holds: The local Clifford moment operator has absolute
spectral gap at least $\gamma/n$ outside the Clifford-invariant
subspace:
\begin{equation}
 \left\|\mathcal M_{\chi_n}^{(t)}
       -\mathcal M_{\overline{\mathcal{C}}_n}^{(t)}\right\|_\infty
 \leq1-\frac{\gamma}{n}.
 \label{eq:walk-clifford-inputs}
\end{equation}
On the Clifford-invariant subspace, the insertion averaged between
Clifford twirls has spectral gap at least $a$ outside the
Haar-invariant subspace. Equivalently,
\begin{equation}
 \mathcal M_{\overline{\mathcal{C}}_n}^{(t)}
 \bigl(I-\mathcal M_{\xi_n}^{(t)}\bigr)
 \mathcal M_{\overline{\mathcal{C}}_n}^{(t)}
 \succeq a \bigl(
     \mathcal M_{\overline{\mathcal{C}}_n}^{(t)}
     -\mathcal M_{\HaarU}^{(t)}\bigr).
 \label{eq:walk-insertion-gap}
\end{equation}
\end{lemma}

After subtracting the Haar moment from a suitable power of the walk's
moment operator, all but constantly many eigenvalues will be
exponentially small in $n$. 
The contribution of the remaining
eigenvalues to an algorithm's acceptance-probability difference is a
finite spectral expansion of the form $u_j=\sum_i b_i\lambda_i^j$,
where $j$ counts repetitions. 
Bounding the
individual coefficients $b_i$ could introduce an unwanted dependence
on the dimension.
Instead, we use that acceptance probabilities are
bounded: the following combinatorial lemma bounds the entire expansion from its
first few values.

\begin{lemma}
\label{lem:walk-scalar-decay}
Let $D\geq1$ and $0<\rho<1$. Suppose
\[
 u_j=\sum_{i=1}^{D}b_i\lambda_i^j\qquad(j\geq1),\quad |\lambda_i|\leq\rho.
\]
If $|u_1|,\ldots,|u_D|\leq B$, then, for every $j\geq D$,
\begin{equation}
 |u_j|\leq B \cdot D \cdot 2^D\binom{j+D-2}{D-1} \cdot \rho^{j-D}.
 \label{eq:walk-scalar-decay}
\end{equation}
\end{lemma}

\begin{proof}
 The generating function of the sequence is
\[
 \sum_{j\geq0}u_{j+1}z^j
 =\frac{N(z)}{q(z)},\qquad
 q(z)=\prod_{i=1}^{D}(1-\lambda_i z),\qquad
 \deg N\leq D-1.
\]
Writing $q(z)=\sum_{k=0}^{D}q_kz^k$, the elementary symmetric
coefficient bound gives $|q_k|\leq\binom Dk\rho^k$. Consequently,
the coefficient of $z^\ell$ in $N(z)$, for $0\leq\ell\leq D-1$, is
\[
 N_\ell=\sum_{k=0}^{\ell}q_k u_{\ell-k+1},
 \qquad |N_\ell|\leq B\sum_{k=0}^{\ell}\binom Dk\rho^k
 \leq B2^D.
\]
Thus $\sum_{\ell=0}^{D-1}|N_\ell|\leq BD2^D$. On the other hand,
expanding the $D$ geometric series in $1/q(z)$ shows that its
coefficient of $z^r$ has absolute value at most
\[
 \binom{r+D-1}{D-1}\rho^r.
\]
For $j\geq D$ and $0\leq\ell\leq D-1$, the index
$r=j-1-\ell$ lies between $j-D$ and $j-1$. Each corresponding
coefficient is therefore bounded by
$\binom{j+D-2}{D-1}\rho^{j-D}$. Taking the coefficient of
$z^{j-1}$ in $N(z)/q(z)$ proves
\eqref{eq:walk-scalar-decay}.
\end{proof}

\subsection{Proof of Theorem~\ref{thm:walk}}
\label{sec:walk-proof}

We will now put everything together to obtain a proof of Theorem~\ref{thm:walk}.
To obtain the desired efficient distinguisher, will use the following stabilizer test. For
$x=(a,b)\in\mathbb F_2^{2n}$, write
$W_x=\bigotimes_i i^{a_ib_i}X^{a_i}Z^{b_i}$ and let
$S(\psi)=\{x:W_x\ket\psi=\pm\ket\psi\}$ be the binary
subspace of Pauli stabilizer labels.

\begin{lemma}[Testing states with many stabilizers]
\label{lem:walk-bell-test}
For $19\leq h\leq n-2$, a uniform polynomial-time test using
$12h+10$ copies accepts every pure state with
$\dim S(\psi)\geq n-h$ with probability at least $5/8$, and
accepts a Haar-random state with probability at most
$1/2+2^{-\Omega(n)}$.
\end{lemma}

This is the Bell-difference test in \cite[Theorem~2 and Supplement,
proof of Theorem~4]{BittelLeone26}. It takes $3h+2$ Bell-difference
samples and computes the subspace of their span commuting with all
samples. If the subspace is nontrivial, it selects a uniform nonidentity
Pauli from it, measures that Pauli on two fresh copies, and accepts
if the outcomes agree; otherwise it outputs a fair bit. This gives
exactly $4(3h+2)+2$ copies, with binary linear algebra for the
classical processing.

\begin{proof}[Proof of \cref{thm:intro-walk}]
Set $p=1/m$ and
\[
 M=\mathcal M_{\nu_n}^{(t)}
  =(1-p)\mathcal M_{\chi_n}^{(t)}
       +p\mathcal M_{\xi_n}^{(t)}.
\]
All these moment operators fix the Haar-invariant subspace. We first
show that $M$ has a gap of order $1/m$ outside that subspace.
For $v$ orthogonal to the image of $\mathcal M_{\HaarU}^{(t)}$, write
$v=x+y$, where
\[
 x=\mathcal M_{\overline{\mathcal{C}}_n}^{(t)}v,\qquad
 y=\bigl(I-\mathcal M_{\overline{\mathcal{C}}_n}^{(t)}\bigr)v.
\]
Since $I-\mathcal M_{\xi_n}^{(t)}$ is positive semidefinite and has
norm at most two, \cref{lem:walk-homeopathy} gives
\[
 \begin{aligned}
 \langle v,(I-\mathcal M_{\xi_n}^{(t)})v\rangle
 &\geq\tfrac12\langle x,(I-\mathcal M_{\xi_n}^{(t)})x\rangle
       -\langle y,(I-\mathcal M_{\xi_n}^{(t)})y\rangle\\
 &\geq\tfrac{a}{2}\|x\|^2-2\|y\|^2.
 \end{aligned}
\]
The first inequality follows by applying
$\|u+w\|^2\geq\frac12\|u\|^2-\|w\|^2$ to
$(I-\mathcal M_{\xi_n}^{(t)})^{1/2}x$ and
$(I-\mathcal M_{\xi_n}^{(t)})^{1/2}y$. Combining it with
\eqref{eq:walk-clifford-inputs} gives
\begin{equation}
 \langle v,(I-M)v\rangle
 \geq\tfrac{a p}{2}\|x\|^2
       +\left((1-p)\frac{\gamma}{n}-2p\right)\|y\|^2
 \geq\tfrac{ ap}{2}\|v\|^2
 \label{eq:walk-moment-gap}
\end{equation}
for sufficiently large $n$, since $p=O_t(n^{-2})$.

On the space $\ker\mathcal M_{\overline{\mathcal{C}}_n}^{(t)}$, the Rayleigh quotient of
$M$ is at most $1-(1-p)\gamma/n$. 
$\ker\mathcal M_{\overline{\mathcal{C}}_n}^{(t)}$ has codimension at most $D_t$.
Then, by the min-max principle at most $D_t$ eigenvalues exceed $1-(1-p)\gamma/n$, which gives $M\succeq-(1-(1-p)\gamma/n)I$, so every negative eigenvalue has absolute
value at most $(1-(1-p)\gamma/n)$.

Consider one length-$m$ segment and subtract its Haar moment, obtaining
$B=M^m-\mathcal M_{\HaarU}^{(t)}$. Equation~\eqref{eq:walk-moment-gap}
and the preceding bounds imply that all eigenvalues of $B$ have
absolute value at most $\rho:=e^{-a/2}<1$ for large $n$.
Except for at most $D_t$ eigenvalues, they have absolute value at most
\begin{equation}
 \eta_n=(1-(1-p)\gamma/n)^m\leq\exp\left(-\frac{\gamma m}{2n}\right).
 \label{eq:walk-fast-modes}
\end{equation}
The Haar eigenvalues of $B$ are zero. Choose the integer $A_t$ large
enough that, with $d=2^n$, $d^{2t}\eta_n\leq e^{-n}$.

Now fix any adaptive $t$-query algorithm $\mathcal A^U$ and measurement
$(M_{\mathcal A},I-M_{\mathcal A})$, and let $\Lambda_{\mathcal A}$ be the
matrix from \cref{lem:walk-query-functional}. Since
$B^j=M^{mj}-\mathcal M_{\HaarU}^{(t)}$, the scalar $a_j=\Tr\!\left[\Lambda_{\mathcal A}^{\dagger}B^j\right]\in [-1,1]$. 
Split its spectral expansion as $a_j=u_j+v_j$,
where $u_j$ contains the contributions of the at most $D_t$
eigenvalues of $B$ larger than $\eta_n$ in absolute value.
By
\cref{lem:walk-query-functional},
\[
 |v_j|\leq d^{2t}\eta_n^j\leq e^{-n}\qquad(j\geq1).
\]
Thus the first $D_t$ values of $u_j$ have absolute value at most two.
Moreover, $u_j$ is a sum of at most $D_t$ terms $b_i\lambda_i^j$
with $|\lambda_i|\leq\rho_t$. Applying
\cref{lem:walk-scalar-decay} gives, for $j\geq D_t$,
\begin{equation}
 |a_j|\leq
 2D_t2^{D_t}\binom{j+D_t-2}{D_t-1}\rho^{j-D_t}+e^{-n}.
 \label{eq:walk-amplification}
\end{equation}
We can now choose $j=3\lceil\log^2 n\rceil$. Since $D_t$ is fixed and
$\rho_t<1$, the polynomial factor in $j$ is dominated by the
exponential decay. Hence $|a_{T/m}|\leq e^{-2\kappa_t\log^2 n}$
for a constant $\kappa_t>0$ and sufficiently large $n$. The algorithm
and measurement were arbitrary, so
\[
 \da t(\nu_n^{*T},\HaarU)\leq e^{-2\kappa_t\log^2 n}.
\]
To distinguish the walk, let $J$ be the number of insertions among
the $T$ updates. Then $J\sim\operatorname{Bin}(T,1/m)$,
$\E[J]=T/m$, and the Chernoff bound gives
\begin{equation}
 \Pr[J>2T/m]\leq(e/4)^{T/m}=e^{-\Omega(\log^2 n)}.
 \label{eq:walk-injection-tail}
\end{equation}
Outside this event there are at most
$h=2T/m=6\lceil\log^2 n\rceil$ single qubit non-Clifford gates. 
At a single qubit Z rotation gate on qubit $i$, the
stabilizers commuting with $Z_i$ remain stabilizers. Commutation with
$Z_i$ imposes at most one binary linear constraint on the stabilizer
labels, so this subspace has codimension at most one. Using $\dim S(\ket{0^n})=n$ gives
\begin{equation}
 \dim S(U\ket{0^n})\geq n-h.
 \label{eq:walk-stabilizer-dimension}
\end{equation}
Prepare
$12h+10=O(\log^2 n)$ copies of $U\ket{0^n}$ by
parallel forward queries and apply \cref{lem:walk-bell-test}.
The walk acceptance is at least $5/8-e^{-\Omega(\log^2 n)}$, while
Haar acceptance is at most $1/2+2^{-\Omega(n)}$.

\end{proof}

\section{Separating higher-order unitary designs from pseudorandomness}
\label{sec:phase}
In this section we will prove~\cref{thm:phase}, which formally separates the notion of pseudorandom unitaries from the statistical notion of polynomial-order unitary designs. 

\subsection{The FB ensemble}
\label{sec:phase-construction}
Our construction in \cref{thm:intro-phase} is based on a product $U=FB$ of two operators:
a block-diagonal operator $B$ that fixes the subspace $\mathrm{span}\{\ket{+}^{\otimes n}\}$ while randomizing
its orthogonal complement, and a diagonal phase function operator $F$ that maps
$\ket{+}^{\otimes n}$ to a phase state.

For the remainder of the section, we let $d=2^n$ denote the dimension of the Hilbert space.

\paragraph{The phase operator.} We begin by describing the phases. Identify the computational basis with the finite field $\F_d$, using a
fixed binary basis, that is, a basis of $\F_d$ as a vector space over $\F_2$.
This representation can be chosen uniformly in
polynomial time by constructing a degree-$n$ irreducible polynomial over
$\F_2$ and using its polynomial basis.\footnote{
Theorem~2.2 in~\cite{shoup1990new} gives a deterministic algorithm with running time
polynomial in the degree and the prime characteristic, which is
polynomial in $n$ here because the characteristic is fixed at two.}
For a uniformly random seed
$a=(a_0,\ldots,a_{2t-1})\in\F_d^{2t}$, define
\begin{equation}
 p_a(X)=\sum_{j=0}^{2t-1}a_jX^j,\qquad
 f_a(x)=\mathrm{tr}_{\F_d/\F_2}(p_a(x)),\qquad
 F_a=\sum_{x\in\F_d}(-1)^{f_a(x)}\proj{x}.
 \label{eq:phase-family}
\end{equation}
Here $\mathrm{tr}_{\F_d/\F_2}(z)=\sum_{j=0}^{n-1}z^{2^j}$ is the
absolute field trace. The seed has $2nt$ bits, and
$\mathcal F_{n,t}$ denotes the distribution of $F_a$.

\begin{lemma}[Statistical independence of the phases]
\label{lem:phase-independence}
Suppose $2t\leq d$. The signs of $F_a$ on any set of at most $2t$
distinct basis labels are independent uniform elements of $\{+1,-1\}$.
For every $t\geq1$, given $a$, the unitary $F_a$ has an exact circuit of size $\poly(n,t)$
using clean work qubits. Moreover, its classical description is computable in
$\poly(n,t)$ time.
\end{lemma}

\begin{proof}
For distinct $x_1,\ldots,x_j$ with $j\leq 2t$, consider the map
\[
 a\longmapsto(p_a(x_1),\ldots,p_a(x_j))
\]
from $\F_d^{2t}$ to $\F_d^j$. Its first $j$ columns form a
Vandermonde matrix with determinant
$\prod_{u<v}(x_v-x_u)\ne0$. It is therefore surjective, and each output
has exactly $d^{2t-j}$ preimages. A uniform seed consequently produces
independent uniform field values at the prescribed labels.
The absolute trace is an $\F_2$-linear map into $\F_2$. It is nonzero:
its defining polynomial has degree $2^{n-1}<d$ and leading coefficient
one, and so cannot vanish on every element of $\F_d$. It follows that
its two fibers both have size $d/2$. Applying this map to the
independent field values proves sign-independence.

Next, we discuss implementations. Horner's rule uses $O(t)$ field operations,
each having Boolean complexity $\poly(n)$ in the fixed polynomial
basis. The trace is binary linear. Reversibly compute $f_a(x)$
using CNOT and Toffoli gates, apply $Z$ to its output bit, and uncompute.
This implements
$\ket{x}\ket0\mapsto(-1)^{f_a(x)}\ket{x}\ket0$
with the stated resources.
\end{proof}

The second factor will be an explicit random circuit fixing
$\ket{+}^{\otimes n}$. Write $N=d-1$.
As an idealized version, let $B_{\infty}$ denote the Haar ensemble on
$\mathrm U\bigl((\ket{+}^{\otimes n})^\perp\bigr)$,
with each unitary extended by the identity on
$\Span\{\ket{+}^{\otimes n}\}$. We use the same identity-extension
convention for special unitaries on this orthogonal complement. Our concrete approach is to approximate the moments of $B_{\infty}$ using polynomial-depth
circuits.

We first analyze $F_aB_\infty$ and then replace the ideal $B_{\infty}$ by
the finite circuit.

\subsection{The $FB_\infty$ ensemble is an approximate design}

We show that $FB_\infty$ fools adaptive queries by controlling
collisions among computational-basis labels. The argument compares
the output of any fixed query algorithm with a common reference state:
both our ensemble and Haar measure contain this state as a positive
component, with the missing weight bounded by a collision probability.
We follow the positive-operator approach of~\cite[Secs.~IV.2--IV.3]{CuiEtAl2025}.

For an ensemble $\nu$ on $\operatorname{U}(d)$, write
\[
    Q_\nu^{(q)}
    =\mathbb E_{C\sim\nu}(C^\dagger)^{\otimes q}
       \Pi_{\mathrm{dist}(d,q)}C^{\otimes q},
    \qquad
    \eta_q(\nu)=\|I-Q_\nu^{(q)}\|_\infty.
\]
$\eta_q(\nu)$ can be interpreted as the largest average collision probability over
all inputs on $q$ copies, and
$(1-\eta_q(\nu))I\preceq Q_\nu^{(q)}\preceq I$.
The corresponding bounds needed for $B_\infty$ and Haar measure are
proved in~\cref{appendix:collision}.

Write $(d)_q=d(d-1)\cdots(d-q+1)$, and
\[
    |\Omega_h\rangle=d^{-q/2}\sum_{x\in[d]^q}|x\rangle_X|x\rangle_Y.
\]
Let $R(\sigma)$ permute the $q$ tensor factors according to
$\sigma\in S_q$, and set
\[
    J_0:=\sum_{\sigma\in S_q}R_X(\sigma)\otimes R_Y(\sigma).
\]
We denote the distinct-subspace projectors on $X$ and $Y$ by
$\Pi_X$ and $\Pi_Y$, respectively.

\begin{lemma}[Positive representation of adaptive queries]
\label{lem:adaptive-positive-representation}
Fix an adaptive $q$-query algorithm, where $1\leq q\leq d$.
There is a positive linear map $\Lambda$, depending only on the
algorithm, such that its averaged output state under an ensemble
$\xi$ is $\rho_\xi=\Lambda(J_\xi)$, where
\[
    J_\xi=d^{2q}\mathbb E_{U\sim\xi}
       (U_X^\dagger)^{\otimes q}|\Omega_{d^q}\rangle\langle\Omega_{d^q}|
       U_X^{\otimes q}.
\]
\end{lemma}

The proof of~\cref{lem:adaptive-positive-representation} is standard and can be found in~\cite{CuiEtAl2025}.
We give a proof consistent with the notation in
Appendix~\ref{app:adaptive-positive-representation}.

We will also use the following reference state
\[
    \rho_*:=\Lambda(J_*),
    \qquad
    J_*:=\frac{d^q}{(d)_q}J_0(I_X\otimes\Pi_Y).
\]
$\rho_*$ is a density operator depending only on the algorithm and can be viewed as the Gaussian average restricted to the distinct subspace. 
Moreover, the Gaussian average is related to the frame operator of the permutation basis, which approximately equals the Haar projector due to the approximate orthogonality of permutations~\cite{harrow2023approximate}. 
The strategy of~\cref{lem:adaptive-collision-bound} is then to reduce the analysis to this Gaussian reference state and then exploit its explicit form.
We include a proof that $\rho_*$ is a density operator in~\cref{app:adaptive-positive-representation}. 
The proof is also contained in different notation in~\cite{CuiEtAl2025}.

Define
\begin{equation}
 Q_\nu^{(q)}=\underset{C\sim\nu}{\E}(C^\dagger)^{\otimes q}
       \Pi_{\mathrm{dist}(d,q)}C^{\otimes q},\qquad
 \eta_q(\nu)=\|I-Q_\nu^{(q)}\|_\infty.
 \label{eq:collision-def}
\end{equation}
The quantity $\eta_q(\nu)$ is the maximum average collision
probability over all joint inputs, including entangled inputs. Since
$0\preceq Q_\nu^{(q)}\preceq I$, it lies in $[0,1]$ and satisfies
\begin{equation}
 (1-\eta_q(\nu))I\preceq Q_\nu^{(q)}\preceq I.
 \label{eq:collision-order}
\end{equation}
For the Haar measure and polynomial $q$, $\eta_q$ is well known to be exponentially small (see e.g.~\cite{10756150}) and the proof for $B_{\infty}$ is an elementary extension of the same argument.
We include a quantitative proof of these bounds in in~\cref{appendix:collision}.

\begin{lemma}[Indistinguishability from collision bounds]
\label{lem:adaptive-collision-bound}
Let $1\leq q\leq d$, let $F$ be a diagonal ensemble with
$2q$-wise independent signs, and let $P$ be a uniform
computational-basis permutation. Then, for every ensemble $\nu$,
\[
    \delta_q(F\ast P\ast\nu,\mu_H)
    \leq\eta_q(\nu)+\eta_q(\mu_H).
\]
\end{lemma}

\begin{proof}
Fix a query algorithm and use~\cref{lem:adaptive-positive-representation}.
For a fixed distinct tuple $x$, the tuple $P^{-1}x$ is uniform
over the $(d)_q$ distinct tuples. Hence
\begin{equation}
    J_{F\ast P}(\Pi_X\otimes\Pi_Y)
    =\frac{d^q}{(d)_q}J_0(\Pi_X\otimes\Pi_Y).
\end{equation}
Conjugating by $(C_X^\dagger)^{\otimes q}$ and averaging over
$C\sim\nu$ yields
\begin{align}
    J_{F\ast P\ast\nu}
    &\succeq (I_X\otimes \Pi_Y) J_{F\ast P\ast\nu} (I_X\otimes \Pi_Y) \\
    &=\frac{d^q}{(d)_q}J_0(Q_\nu^{(q)}\otimes\Pi_Y)\\
    &\succeq(1-\eta_q(\nu))J_*.
\end{align}
Here, the last inequality holds by~\eqref{eq:collision-order}.

Applying the positive map $\Lambda$ gives
$\rho_{F\ast P\ast\nu}\succeq(1-\eta_q(\nu))\rho_*$.
Since both states have trace one, this implies
\[
    \frac12\|\rho_{F\ast P\ast\nu}-\rho_*\|_1
    \leq\eta_q(\nu).
\]
Taking $C$ Haar random gives the same comparison with error
$\eta_q(\mu_H)$, because $FPC$ is then Haar random.
The lemma then follows from applying the triangle inequality. 
\end{proof}

We now apply the criterion to the fixed-vector ensemble.

\begin{lemma}[Indisinguishability for $FB_\infty$]
\label{lem:ideal-fb-design}
For $1\leq q\leq t$ and $2t\leq d$, the phase ensemble
$F_{n,t}$ from~\cref{sec:phase-construction} satisfies
\[
    \delta_q(F_{n,t}\ast B_\infty,\mu_H)
    \leq\frac{q(q-1)}{d}+\frac{q(q-1)}{d+1}
    \leq\frac{2q(q-1)}{d}.
\]
\end{lemma}

\begin{proof}
Every computational-basis permutation fixes $|+\rangle^{\otimes n}$,
so $PB_\infty$ has the same distribution as $B_\infty$.
The phases are $2q$-wise independent by~\cref{lem:phase-independence}.
Apply~\cref{lem:adaptive-collision-bound} and the collision
bounds of~\cref{lem:collision}.
\end{proof}

\subsection{Learning the planted phase state}
\label{sec:learning}
Here, we will prove the following lemma:
\begin{lemma}\label{lemma:efficient-distinguisher}
    Any ensemble of the form $F B$ with $F$ is a random phase drawn as $F_a$ in~\cref{sec:phase-construction} and $B$ leaving the $\ket{+^n}$ state invariant can be distinguished from Haar in polynomial time using parallel queries.
\end{lemma}

Every forward query on $\ket{+}^{\otimes n}$ produces the same phase
state. 
We learn this state by extracting relative phases, each of which gives a linear constraint on its seed. 
\subsubsection{Relative-phase measurements and recovery}

We use the derivative measurement of Arunachalam, Bravyi, Dutt,
and Yoder~\cite[Definition~1]{arunachalam2023optimalalgorithmslearningquantum},
with a change of binary coordinates to randomize its direction.

\begin{lemma}[Derivative sampling]
\label{lem:derivative-sample}
Let $d=2^n$, $n\geq1$, and $f:\F_d\to\F_2$. For every integer
$J\geq1$, a quantum procedure using at most one copy of
$\ket{\psi_f}=d^{-1/2}\sum_x(-1)^{f(x)}\ket{x}$ and
$\poly(n,J)$ operations either aborts, with probability $d^{-J}$,
or returns a uniform unordered distinct pair $\{x,y\}$ and the
bit $f(x)+f(y)$.
\end{lemma}

\begin{proof}
Sample $J$ independent uniform $\F_d$ elements and take the first
nonzero one as $h$, and abort if all are zero, which happens with probability $d^{-J}$.
Since $h\neq0$, it extends to a binary basis $h,h_2,\ldots,h_n$ of
$\F_d$. Every $x\in\F_d$ is then uniquely of the form $x=bh+y_z$ with
$b\in\F_2$, $z=(z_2,\ldots,z_n)\in\F_2^{n-1}$, and
$y_z=\sum_{i=2}^n z_ih_i$. Gaussian elimination finds the coordinate
map $bh+y_z\mapsto(b,z)$, which is invertible and $\F_2$-linear on
$n$-bit strings and is therefore implemented by $O(n^2)$ CNOT gates.
After this change of coordinates,
measure the last $n-1$ qubits in the computational basis. Each outcome
$z$ has probability $2/d$, and the remaining qubit is
\[
 \frac{(-1)^{f(y_z)}\ket0+(-1)^{f(y_z+h)}\ket1}{\sqrt2}.
\]
Measuring it in the $X$ basis alway gives $f(y_z)+f(y_z+h)$.
Return $\{y_z,y_z+h\}$ and $f(y_z)+f(y_z+h)$. Every distinct pair
$\{x,y\}$ has exactly one direction $h=x+y$ and one coset for
that direction. Its probability is therefore
$(d-1)^{-1}(2/d)=\binom d2^{-1}$. All preprocessing takes
polynomial time and is performed before using the input copy.
\end{proof}

For the analysis, write $p_a(X)=\sum_{j=0}^{L-1}a_jX^j$ and
$f_a(x)=\mathrm{tr}_{\F_d/\F_2}(p_a(x))$, for $L=2t$. Different seeds can represent the same trace
function, and a constant difference changes the state only by a global
sign. All other differences are separated by the following estimate.

\begin{lemma}[Separation of the phase family]
\label{lem:separation}
Let $d=2^n\geq16$ and $2\leq L\leq d^{1/4}$. If
$a,a'\in\F_d^L$ and $g=f_a+f_{a'}$ is nonconstant, then
a uniform unordered distinct pair satisfies
\begin{equation*}
 \Pr[f_a(x)+f_a(y)=f_{a'}(x)+f_{a'}(y)]\leq\frac58.
\end{equation*}
\end{lemma}

\begin{proof}
We use the additive-character Weil bound in characteristic two:
if $P\in\F_d[X]$ has odd degree $e$, then\footnote{This is the case
$p=2$ of Theorem~5.38 in~\cite{lidl1997finite}, applied to the
character $z\mapsto(-1)^{\mathrm{tr}(z)}$, which is nontrivial
because the trace is surjective.}
$|\sum_{x\in\F_d}(-1)^{\mathrm{tr}(P(x))}|\leq(e-1)\sqrt d$.
To meet the degree hypothesis, we first remove the even exponents
from $p_a+p_{a'}$ without changing its trace. Frobenius invariance
of the trace and $x^{2^n}=x$ give
\begin{equation*}
 \mathrm{tr}(cx^{2m})
 =\mathrm{tr}\!\left((cx^{2m})^{2^{n-1}}\right)
 =\mathrm{tr}\!\left(c^{2^{n-1}}x^m\right).
\end{equation*}
Repeatedly apply this identity to each even nonconstant exponent and
collect equal powers. The resulting polynomial $P$ has degree less
than $L$, only odd nonconstant exponents, and
$g(x)=\mathrm{tr}(P(x))$. Since $g$ is nonconstant, $P$ has
positive odd degree $e$.

The Weil bound applied to $P$ thus gives, for
$\mu_g=d^{-1}\sum_x(-1)^{g(x)}$,
\[
 |\mu_g|\leq(e-1)/\sqrt d
 \leq(L-2)/\sqrt d\leq d^{-1/4}\leq\tfrac12.
\]
It remains to pass from single values of $g$ to pairs. Since the sums
are in $\F_2$, the event
$f_a(x)+f_a(y)=f_{a'}(x)+f_{a'}(y)$ is exactly the event
$g(x)=g(y)$. Set $s_x=(-1)^{g(x)}\in\{\pm1\}$, so that
$\mathbf1_{g(x)=g(y)}=(1+s_xs_y)/2$. We sum over all ordered pairs
gives $\sum_{x,y}s_xs_y=(\sum_xs_x)^2=d^2\mu_g^2$, and each of the
$d$ diagonal terms is $s_x^2=1$, so
$\sum_{x\ne y}s_xs_y=d^2\mu_g^2-d$. Averaging over the $d(d-1)$
ordered pairs with $x\ne y$, which induces the uniform distribution on
unordered distinct pairs, yields
\[
 \Pr_{x\ne y}[g(x)=g(y)]
 =\frac12\left(1+\frac{d\mu_g^2-1}{d-1}\right)
 \leq\frac{1+\mu_g^2}{2}\leq\frac58,
\]
where the first inequality is equivalent to $\mu_g^2\leq1$ and the
second uses $|\mu_g|\leq\frac12$.
\end{proof}

\begin{proposition}[Learning up to global phase]
\label{prop:learning}
Let $d=2^n\geq16$ and $2\leq L\leq d^{1/4}$. For every
$0<\delta<1$, a uniform quantum algorithm uses at most
\begin{equation*}
 K=O(nL+\log(1/\delta))
\end{equation*}
copies of $\ket{\psi_a}$ and returns $\widehat a\in\F_d^L$
with $\ket{\psi_{\widehat a}}=\pm\ket{\psi_a}$ with probability
at least $1-\delta$. Its worst-case running time is
$\poly(n,L,\log(1/\delta))$. On arbitrary inputs it either reports
failure or returns a seed from this same family.
\end{proposition}

\begin{proof}
Fix the binary basis $e_1,\ldots,e_n$ of $\F_d$ and write
$a_j=\sum_{\ell=1}^n a_{j\ell}e_\ell$. Apply
\cref{lem:derivative-sample} independently to the copies, using
$J=\lceil\log(2K/\delta)/\log d\rceil$ trials per direction and
reporting failure on any abort. The probability of an abort is at most
$Kd^{-J}\leq\delta/2$. Conditional on not aborting, the resulting
pairs remain independent and uniform.

Each sample $(x_i,y_i,b_i)$ gives the binary linear equation
\begin{equation*}
 \sum_{j=0}^{L-1}\sum_{\ell=1}^n
 a_{j\ell}\,\mathrm{tr}\!\left(e_\ell(x_i^j+y_i^j)\right)=b_i.
\end{equation*}
Form their matrix $W\in\F_2^{K\times nL}$ and solve
$W\widehat a=b$ by Gaussian elimination, identifying seeds with
their binary coordinate vectors. Reject an inconsistent system and
otherwise return any solution. Field arithmetic computes every entry
of $W$ in polynomial time.

On the promised input the solution set is $a+\ker W$. Every row
annihilates the binary subspace
$\{c\in\F_d^L:f_c\text{ is constant}\}$, whereas
\cref{lem:separation} implies that any fixed $c\in\F_d^L$ with
$f_c$ nonconstant is annihilated by all $K$ independent rows with
probability at most $(5/8)^K$. Hence
\begin{equation*}
 \Pr[\ker W\ne\{c\in\F_d^L:f_c\text{ is constant}\}]
 \leq\sum_{\substack{c\in\F_d^L\\f_c\text{ nonconstant}}}\Pr[Wc=0]
 \leq2^{nL}(5/8)^K\leq\delta/2.
\end{equation*}
Otherwise every solution is $\widehat a=a+c$ for some
$c\in\F_d^L$ with $f_c$ constant, and
$\ket{\psi_{\widehat a}}=(-1)^{f_c(0)}\ket{\psi_a}$.
Adding the abort probability gives total
failure at most $\delta$. 
\end{proof}

\subsection{Proof of~\cref{thm:phase}}
We are now in a position to replace the $B_{\infty}$ ensemble by an approximate $t$-design in the fixed-vector subpsace. This will conclude the proof of Theorem~\ref{thm:intro-phase}.
We use the machinery from~\cite{anshu2026depth}. 
For a compact group $G$ and a subset $S\subseteq G$, the
\emph{Kazhdan constant} is
\[
    \mathcal{K}(G;S)
    =\inf_{\rho}\inf_{\|v\|=1}\sup_{g\in S}\|\rho(g)v-v\|,
\]
where $\rho$ ranges over continuous unitary representations of $G$
with no nonzero invariant vectors, and $v$ lies in the representation
space. Generating sets are understood to generate $G$ topologically.
We will use multiple standard results about Kazhdan constants (compare~\cite{kassabov2007symmetric}):
\begin{lemma}[Bounding the largest eigenvalue with the Kazhdan constant {\cite[Lemma 2.17]{chen2024incompressibilityspectralgapsrandom}}]\label{lem:kazhdan_gaps_relation}
    Let $G$ be a compact group (finite or Lie) topologically generated by a finite subset $S$ which is closed under taking inverses. Then for any finite-dimensional continuous unitary representation $\rho$ of $G$, 
    \begin{equation}\label{eq:kazhdan_gap_max_eval}
        \lambda_{\max}\left(\frac{1}{\abs{S}}\sum_{g\in S} \rho(g) - \E_{\nu\sim \mu(G)} \rho(\nu)\right) \leq 1 - \frac{\mathcal{K}(G;S)^2}{2\abs{S}}, 
    \end{equation}
    where $\lambda_{\max}(\cdot)$ denotes the largest eigenvalue. If $S$ contains the identity, then $\cref{eq:kazhdan_gap_max_eval}$ holds with $\lambda_{\max}(\cdot)$ replaced with the operator norm $\norm{\cdot}$. 
\end{lemma}
Moreover, we will use the following reduction lemma:
\begin{lemma}[\cite{kassabov2007symmetric}]\label{lemma:kazhdan-by-subgroups}
    Let $H_i$ be subgroups of $G$ and $S_i\subseteq H_i$.
    Then, we have
    \begin{equation}
        \mathcal{K}\left(G;\bigcup_i S_i\right)\geq \frac{1}{2} \mathcal{K}\left(G;\bigcup_i H_i\right) \inf_i \mathcal{K}\left(H_i;S_i\right).
    \end{equation}
\end{lemma}
Finally, we use the main result from~\cite{anshu2026depth} (see also~\cite{baer2026random}):
\begin{theorem}\label{thm:expander}
There is a set of unitaries $S=\{V_1,...V_{25}\}$ on $n$ qubits such that each $V_j$ is a constant depth circuit and $S$ has a uniformly lower bounded Kazhdan constant in $\mathrm{SU}(2^n)$. 
\end{theorem}
The aim of this section will be to combine these ingredients with the following decomposition of $\mathrm{SU}\bigl((\ket{+}^{\otimes n})^\perp\bigr)$ that we prove in Appendix~\ref{appendix:decomposition}.
\begin{lemma}\label{lemma:decomp}
Let $n\geq 1$, for $i\in\{1,\ldots,n\}$ and $V\in\mathrm{SU}(2^{n-1})$, write
\begin{equation}
  \mathsf{C}_i(V)
  =|+\rangle\langle +|_i\otimes I
   +|-\rangle\langle -|_i\otimes V,
\end{equation}
where $V$ acts on all qubits other than $i$.
Every $U\in \mathrm{SU}\bigl((\ket{+}^{\otimes n})^\perp\bigr)$ admits a factorization into $2n-1$ such gates,
\begin{equation}
  U=\mathsf{C}_1(U_1)\cdots\mathsf{C}_{n-1}(U_{n-1})
    \mathsf{C}_n(U_n)
    \mathsf{C}_{n-1}(U_{n+1})\cdots\mathsf{C}_1(U_{2n-1}).
\end{equation}
\end{lemma}

\paragraph{Putting everything together.} We will now complete the proof of Theorem~\ref{thm:phase}.

\begin{proof}[Proof of Theorem~\ref{thm:phase}]
Define $H_i:=\mathsf{C}_i(\mathrm{SU}(2^{n-1}))$.
 As $\mathcal{K}(G;G)\geq \sqrt{2}$~\cite{kassabov2007symmetric} for any compact group $G$, we can apply the short-product lemma~\cite{chen2024incompressibilityspectralgapsrandom} to obtain that $\mathcal{K}(\mathrm{SU}\bigl((\ket{+}^{\otimes n})^\perp\bigr);\bigcup_{j=1}^{2n-1}H_j)=\Omega(n^{-1})$ (where we identify $n+j$ with $n-j$). 
 We can then apply~\cref{lemma:kazhdan-by-subgroups} for the sets $S_j:=\mathsf{C}_j(S)$.
 As $\mathsf{C}_j(S)$ has a uniformly lower bounded Kazhdan constant in $\mathsf{C}_j(\mathrm{SU}(2^{n-1}))$ by Theorem~\ref{thm:expander}, the finite inverse-closed generating set
\[
    \hat{S}:=\{I\}\cup\bigcup_{j=1}^{n}(S_j\cup S_j^{-1})
\]
contains the identity and satisfies $\mathcal{K}(\mathrm{SU}\bigl((\ket{+}^{\otimes n})^\perp\bigr), \hat{S})=\Omega(n^{-1})$ with $|\hat{S}|=O(n)$.
 Finally, for the spectral gap of the uniform probability distribution on this generator set that yields $\Delta=\Omega(n^3)$ by~\cref{lem:kazhdan_gaps_relation}.
 Denote by $\nu_{\hat{S}}$ the uniform probability distribution over elements in $\hat{S}$.
 The gap can be easily amplified by taking $\nu_{\hat{S}}^{*r}$ to yield $\varepsilon$-approximate designs in depth $r=O(\frac{1}{\Delta}(2nt+\log(1/\varepsilon)))=O(n^3(2nt+\log(1/\varepsilon))$ (see e.g.~\cite{Brand_o_2016}).
 In particular, choosing $\varepsilon=e^{-n}$ we end up with polynomially deep circuits. 

 We can then plug $B_r:=\nu_{\hat{S}}^{*r}$ into the position of $B_{\infty}$ in $FB_{\infty}$. 
 As $FB_{\infty}$ is indistinguishable from Haar for polynomial $t$ by~\cref{lem:ideal-fb-design} and $\nu_{\hat{S}}^{*r}$ is indistinguishable from $B_{\infty}$ using $t$ copies via the above we obtain the desired construction.
 In particular, the ensemble is over circuits of polynomial depth for every fixed polynomial $t=\mathrm{poly}(n)$ it yields an approximate $t$-design. 
 However, because of~\cref{lemma:efficient-distinguisher} the ensemble can be efficiently distinguished from Haar. This means, the ensemble fails to be pseudorandom according to Definition \ref{def:PRU}. This completes the proof of~Theorem \ref{thm:phase}.
\end{proof}

\section{Pseudorandomness and Maximal Scrambling}
\label{sec:scrambling}

In this section, we explore how genuine pseudorandomness can emerge in simple quantum processes by drawing on an explicit connection to many-body physics.

Local random quantum circuits are among the simplest models of time-dependent quantum dynamics, where information initially stored in a small subsystem spreads globally through interactions with the rest of the system. This process, known as \emph{scrambling}, is closely related to \emph{thermalization}: as information spreads, it becomes increasingly difficult to recover from local observations~\cite{Hayden_Preskill_2007,Liu_2018}. For local random circuits, matching progressively higher moments of the Haar measure requires increasing \emph{circuit depth}~\cite{Brand_o_2016}, making the unitary design order a natural way to describe how randomness develops over time.

Liu, Lloyd, Zhu, and Zhu~\cite{Liu_2018}  showed that higher and higher moment matching (in the form of unitary $t$-designs) implies a saturation in the hierarchy of entanglement entropies. Eventually, this process reaches a limit: when $t= \Theta(n)$ all entanglement entropies are saturated, which is known as \emph{maximal scrambling} (see Figure~\ref{fig:scrambling-hierarchy}). At this point, the system appears effectively \emph{thermal}: initially localized information has fully dispersed and becomes inaccessible to local probes.
Maximal scrambling therefore suggests a physically motivated regime in which moment matching might be sufficient for genuine pseudorandomness. Once all these entanglement entropies have saturated, it seems plausible that an efficient observer should have little remaining structure to exploit. There is a closely related intuition behind the HMMR conjecture and the design of block ciphers: \emph{diffusion} spreads changes in a few input bits across many output bits, while sufficiently strong statistical independence limits the correlations exploited by differential and linear attacks~\cite{HMMR2005,noauthororeditor}. Connections between these attacks and scrambling have also been studied quantitatively~\cite{CHAMON2022169086}. In the quantum setting, this motivates the following question:
\begin{center}
\emph{Does maximal scrambling imply computational pseudorandomness?}
\end{center}

\begin{figure}[t]
\centering
\begin{tikzpicture}[
>=Latex,
axis/.style={line width=0.8pt},
level/.style={line width=0.8pt},
dashlevel/.style={line width=0.8pt, dashed, dash pattern=on 4pt off 4pt},
lab/.style={font=\small},
smalllab/.style={font=\footnotesize}
]

\def\xaxis{0}
\def\xright{4.2}
\def\xlabel{4.75}

\def\yone{0.35}
\def\ydash{0.70}
\def\ytwo{1.00}
\def\ythree{1.60}
\def\ylog{3.05}
\def\ytop{4.20}

\draw[axis,->] (\xaxis,0.15) -- (\xaxis,4.55);
\node[lab,above] at (\xaxis,4.55) {Design order};

\node[lab,left] at (\xaxis-0.12,\yone) {$1$};
\node[lab,left] at (\xaxis-0.12,\ytwo) {$2$};
\node[lab,left] at (\xaxis-0.12,\ythree) {$3$};
\node[lab,left] at (\xaxis-0.12,\ylog) {$t=\Theta(n)$};

\draw[level] (\xaxis,\yone) -- (0.15,\yone);

\draw[dashlevel] (\xaxis,\ydash) -- (\xright,\ydash);
\node[lab,right] at (\xlabel,\ydash) {Min-scrambling};

\draw[level] (\xaxis,\ytwo) -- (\xright,\ytwo);
\draw[level] (\xaxis,\ythree) -- (\xright,\ythree);


\node[lab] at (2.05,2.20) {$\cdots$};

\draw[level] (\xaxis,\ylog) -- (\xright,\ylog);
\node[lab,right] at (\xlabel,\ylog) {Max-scrambling};

\fill[gray!18] (\xaxis,\ylog) rectangle (\xright,\ytop);
\node[lab] at (2.05,3.65) {Haar-like entropies};

\draw[level] (\xaxis,\ylog) -- (\xright,\ylog);
\draw[level,gray!35] (\xaxis,\ytop) -- (\xright,\ytop);
\draw[level,gray!35] (\xright,\ylog) -- (\xright,\ytop);

\end{tikzpicture}
\caption{The hierarchy of scrambling and unitary design order, following~\cite{Liu_2018}.
At $t=\Theta(n)$, moment matching achieves a \emph{saturation} via near-maximal expected min-entropy across a fixed balanced crossed cut, and hence near-maximal expected entropies at every positive R\'enyi order.}
\label{fig:scrambling-hierarchy}
\end{figure}
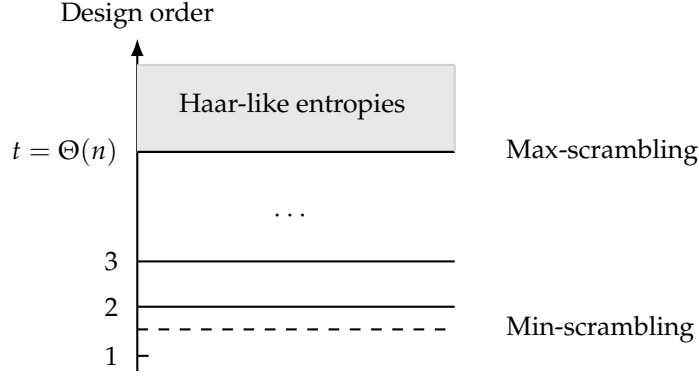

Our results suggest caution; they show that the answer is still \emph{no}, at least in the general case. As a byproduct of our separation between unitary designs and pseudorandom unitaries in \Cref{thm:phase}, we find that the $FB$ ensemble can yield an (approximate) unitary $t$-design in the maximal scrambling regime $t=\Theta(n)$, and still admit enough structure to be efficiently distinguishable from Haar. This means that maximally scrambled systems need not, in general be pseudorandom.

To state this more precisely, let $n$ be even and divide both the input reference $AB$ and the output $CD$ into equal halves. We can associate each unitary $U$ with its normalized Choi state
\begin{equation}
\ket{U}_{ABCD}
=2^{-n/2}\sum_{x\in\bit^n}\ket{x}_{AB}\otimes U\ket{x}_{CD},
\label{eq:intro-choi-state}
\end{equation}
and let $\rho_{AC}(U)=\Tr_{BD}(\proj{U})$. Entanglement across $AC:BD$ measures how the input subsystems have been mixed among the output subsystems. The min-entropy $S_{\min}(\rho)=-\log_2\norm{\rho}_\infty$ has maximal value $n$ for this reduced state.

Using \Cref{thm:phase} and the results in~\cite{Liu_2018}, we can conclude the following\footnote{Technically, the entropy saturation results for $t=\Theta(n)$ in~\cite{Liu_2018} only hold for exact unitary $t$-designs; however, we found that the argument can be extended to our approximate designs as the error is negligibly small.}: there exists a polynomial-depth circuit ensemble ensemble $\nu = \{\nu_n\}$ that forms an approximate unitary $t$-design with $t=\Theta(n)$ and sufficiently small error such that, for sufficiently large even $n$,
\[
 \underset{U\sim\nu_n}{\E} S_{\min}\bigl(\rho_{AC}(U)\bigr)\geq n-8
\]
for every fixed balanced crossed cut $AC:BD$. Nevertheless, an efficient quantum algorithm can distinguish $\nu$ from Haar measure with constant advantage using $O(n^2)$ forward queries.

To obtain this result, it suffices to use the $FB$ construction with $t=\Theta(n)$ and choose the mixing circuit $B$ deep enough to control Choi moments of order $\sim 4n$. This gives the expected min-entropy bound above. Importantly, increasing the depth of $B$ preserves the relation $U\ket{+}^{\otimes n}=F\ket{+}^{\otimes n}$, so the same phase-learning algorithm continues to distinguish the ensemble. Thus the ensemble exhibits near-maximal entanglement entropies while preserving algebraic structure that an efficient observer can learn. This suggests some caution when using unitary designs to model scrambling in many-body physics and black-hole evolution~\cite{Hayden_Preskill_2007,Liu_2018,poremba_et_al:LIPIcs.ITCS.2026.109}, as even maximally scrambled systems can exhibit non-trivial structure which is accessible to efficient experiments. Moreover, it suggests that pseudorandomness can serve as much more reliable proxy for Haar randomness, providing a new frontier for complexity beyond maximal scrambling.

At the same time, our results suggest a new conjecture for how pseudorandomness may very well emerge in simple quantum processes. Matching moments up to $t=\Theta(n)$ has a clear physical motivation through maximal scrambling, and our separation at this order relies on a artificial non-local structured ensemble.
This leaves open a promising intermediate regime, where we retain the structure of independent \emph{local} interactions to ensure some amount of "simplicity" while also matching the statistical moments up to the maximal scrambling scale. Motivated by these principles, we therefore put forward the following new conjecture.

\begin{conjecture}[Pseudorandomness in the Maximal Scrambling Regime]
Let $\{\nu_n\}_n$ be an efficiently samplable ensemble of local quantum gates on $n$ qubits, and let $\nu_n^{*T}$ denote the ensemble obtained by independently composing $T=\poly(n)$ gates drawn from $\nu_n$.
If $\{\nu_n^{*T}\}_n$ forms an approximate unitary $t$-design for $t=\Theta(n)$ with negligible error, then $\{\nu_n^{*T}\}_n$ is a pseudorandom unitary ensemble.
\end{conjecture}

As in the quantum HMMR conjecture, each gate acts non-trivially on only a constant number of qubits, and the distribution $\nu_n$ may depend on $n$. Importantly, neither of our counterexamples applies to this setting: the near-Clifford argument establishes designs only at fixed constant order, and the $FB$ ensemble uses a separate, artificial phase circuit and a structured nonlocal mixing. 

Known results on local random quantum circuits already give polynomial-depth designs at $t=\Theta(n)$~\cite{Brand_o_2016}.
Haferkamp and Hunter-Jones~\cite{hunterjones2026idealrandomquantumcircuits} recently showed that their outcome probabilities become uniformly small at order $t=\Theta(n)$, which can be seen as further evidence for this threshold.
More generally, it suggests studying growing-order independence in local reversible circuits as a possible way to make progress beyond the classical HMMR counterexamples. We hope that the study of this connection to scrambling and entanglement saturation provides a new research direction for the Gowers conjecture in both the classical and the quantum setting.

\bibliographystyle{alpha}
\bibliography{main}

\appendix

\section{Bounds on the collision probability}\label{appendix:collision}
For $ q\leq d$, recall the definitions
\begin{equation}
 Q_\nu^{(q)}=\underset{C\sim\nu}{\E}(C^\dagger)^{\otimes q}
       \Pi_{\mathrm{dist}(d,q)}C^{\otimes q},\qquad
 \eta_q(\nu)=\|I-Q_\nu^{(q)}\|_\infty.
 \label{eq:collision-def-appendix}
\end{equation}
The quantity $\eta_q(\nu)$ is the maximum average collision
probability over all joint inputs, including entangled inputs. Since
$0\preceq Q_\nu^{(q)}\preceq I$, it lies in $[0,1]$ and satisfies
\begin{equation}
 (1-\eta_q(\nu))I\preceq Q_\nu^{(q)}\preceq I.
 \label{eq:collision-order-appendix}
\end{equation}

We will use the following standard consequence of the Schur--Weyl duality (see e.g.~\cite{harrow2013church}).

\begin{lemma}[Two-copy symmetric twirl]
\label{lem:two-copy-symmetric-twirl}
For $N\geq1$, let $\Pi_{\mathrm{sym}}$ project onto the symmetric
subspace of $(\C^N)^{\otimes2}$. If
$X=\Pi_{\mathrm{sym}}X\Pi_{\mathrm{sym}}$, then
\begin{equation}
 \underset{V\sim\HaarU}{\E}V^{\otimes2}X(V^\dagger)^{\otimes2}
 =\frac{\Tr X}{N(N+1)/2}\Pi_{\mathrm{sym}}.
 \label{eq:two-copy-symmetric-twirl}
\end{equation}
\end{lemma}
\begin{lemma}[Collision probability]
\label{lem:collision}
For $d=2^n$, $n\geq1$, and $1\leq q\leq d$, Haar measure
$\Binf$ on the group fixing $\ket{+}^{\otimes n}$ satisfies
\begin{equation}
 \eta_q(\Binf)\leq\frac{q(q-1)}d,\qquad
 \eta_q(\HaarU)\leq\frac{q(q-1)}{d+1}.
 \label{eq:collision-bound}
\end{equation}
The Haar bound also holds in every integer dimension $D\geq2$,
with $d$ replaced by $D$.
\end{lemma}
\begin{proof}
Let $\Pi_{=}=\sum_x\proj{x,x}$,
$\mathcal S=\Span\{\ket{+}^{\otimes n}\}$,
and $N=d-1$. The vectors
$\ket{k_x}=\ket{x}-d^{-1/2}\ket{+}^{\otimes n}$ satisfy
$\sum_x\proj{k_x}=I_{(\ket{+}^{\otimes n})^\perp}$ and $\|k_x\|^2=1-1/d$.
Moreover,
\begin{equation}
 \ket{x,x}=\frac1d\ket{+}^{\otimes2n}
 +\frac1{\sqrt d}\bigl(\ket{+}^{\otimes n}\otimes\ket{k_x}
       +\ket{k_x}\otimes\ket{+}^{\otimes n}\bigr)
 +\ket{k_x}^{\otimes2}.
 \label{eq:collision-decomposition}
\end{equation}
The subgroup acts as $1\oplus V$ on $\mathcal S\oplus(\ket{+}^{\otimes n})^\perp$.
Averaging over its central phases $1\oplus e^{i\theta}I_{(\ket{+}^{\otimes n})^\perp}$
removes cross terms between the three displayed summands. On
$\mathcal S^{\otimes2}$ the averaged collision operator has
eigenvalue $1/d$. On
$(\mathcal S\otimes(\ket{+}^{\otimes n})^\perp)\oplus((\ket{+}^{\otimes n})^\perp\otimes\mathcal S)$,
identified with $\C^2\otimes(\ket{+}^{\otimes n})^\perp$, it is
\[
 \frac1d\begin{pmatrix}1&1\\1&1\end{pmatrix}\otimes I_{(\ket{+}^{\otimes n})^\perp},
\]
which is already invariant and has norm $2/d$.
On $\bigl((\ket{+}^{\otimes n})^\perp\bigr)^{\otimes2}$, \cref{lem:two-copy-symmetric-twirl}
gives a multiple of the symmetric projector with coefficient
\[
 \frac{\sum_x\|k_x\|^4}{N(N+1)/2}
 =\frac{2(d-1)}{d^2}\leq\frac2d.
\]
Hence
\begin{equation}
 \underset{B_\infty\sim\Binf}{\E}(B_\infty^\dagger)^{\otimes2}
          \Pi_{=}B_\infty^{\otimes2}\preceq\frac2d I.
 \label{eq:two-copy-collision}
\end{equation}
For full Haar measure, $\Pi_{=}$ has symmetric support and trace
$d$, so the same lemma gives $2\Pi_{\mathrm{sym}}/(d+1)$.
Finally,
$I-\Pi_{\mathrm{dist}(d,q)}\preceq\sum_{i<j}\Pi_{=}^{(i,j)}$,
with identities on the other copies: every nondistinct tuple has an
equal pair. Conjugating and averaging
this inequality, then summing the $\binom q2$ two-copy bounds,
proves both assertions. For $q=1$ the sum is empty.
\end{proof}
\section{Proof of~\cref{lemma:decomp}}\label{appendix:decomposition}

\begin{lemma}[Controlled decomposition of fixed-vector special unitaries]
\label{lem:su-fixed-vector-palindrome}
For $i\in\{1,\ldots,n\}$ and $V\in\operatorname{U}(2^{n-1})$, write
\[
    \mathsf{C}_i(V)=|+\rangle\langle +|_i\otimes I
           +|-\rangle\langle -|_i\otimes V,
\]
where $V$ acts on all qubits other than $i$.
Every $U\in\operatorname{SU}\bigl((\ket{+}^{\otimes n})^\perp\bigr)$ admits a factorization
\[
    U=\mathsf{C}_1(L_1)\cdots \mathsf{C}_{n-1}(L_{n-1})\mathsf{C}_n(V_n)
      \mathsf{C}_{n-1}(R_{n-1})\cdots \mathsf{C}_1(R_1),
\]
with all $L_k,R_k,V_n\in\operatorname{SU}(2^{n-1})$.
\end{lemma}

\begin{proof}
We first work in the computational basis, writing
\[
    \mathsf{C}^{(Z)}_i(V)=|0\rangle\langle0|_i\otimes I
                         +|1\rangle\langle1|_i\otimes V.
\]
Define
\[
    F_k=(\mathbb C^2)^{\otimes(k-1)}
        \otimes\operatorname{span}\{|0^{n-k+1}\rangle\}.
\]
We show that any $M_k\in\operatorname{SU}(2^n)$ fixing $F_k$
pointwise can be written as
\[
    M_k=\mathsf{C}^{(Z)}_k(L_k)M_{k+1}\mathsf{C}^{(Z)}_k(R_k),
\]
where $L_k,R_k\in\operatorname{SU}(2^{n-1})$ and
$M_{k+1}\in\operatorname{SU}(2^n)$ fixes $F_{k+1}$ pointwise.

For $k<n$, let $A_k$ be the subspace with qubit $k$ equal to $1$,
and let $D_k$ be the subspace with qubit $k$ equal to $0$ and
$(x_{k+1},\ldots,x_n)\neq 0^{n-k}$. Then
$F_k^\perp=A_k\oplus D_k$. Write $A_k=E_k\oplus A'_k$, where
\[
    E_k=(\mathbb C^2)^{\otimes(k-1)}
        \otimes\operatorname{span}\{|1\rangle|0^{n-k}\rangle\}.
\]
We have $\dim E_k=2^{k-1}$ and
$\dim A'_k=\dim D_k=2^{n-1}-2^{k-1}>0$.
The cosine--sine decomposition~\cite{golub2013matrix}, with the identity block placed on
$E_k$, gives
\[
    M_k\big|_{F_k^\perp}
    =(\widetilde L_k\oplus P_k)
    \begin{pmatrix}
        I_{E_k}&0&0\\
        0&C&-S\\
        0&S&C
    \end{pmatrix}
    (\widetilde R_k\oplus Q_k),
\]
in the order $E_k\oplus A'_k\oplus D_k$, where $C,S$ are diagonal
and $C^2+S^2=I$. Absorbing $P_k,Q_k$ into the middle factor and
extending by the identity on $F_k$ yields
\[
    M_k=\mathsf{C}^{(Z)}_k(\widetilde L_k)N_{k+1}\mathsf{C}^{(Z)}_k(\widetilde R_k),
\]
where $N_{k+1}$ fixes $F_k\oplus E_k=F_{k+1}$ pointwise.

We now enforce determinant one. Since $A'_k\neq\{0\}$, choose
unitaries $X_k,Y_k$ on $A_k$, acting identically on $E_k$, such that
\[
    \det X_k=(\det\widetilde L_k)^{-1},
    \qquad
    \det Y_k=(\det\widetilde R_k)^{-1}.
\]
For example, each correction can be a phase on one fixed unit vector
in $A'_k$, with the identity on its orthogonal complement. Set
\[
    L_k=\widetilde L_kX_k,\qquad
    R_k=Y_k\widetilde R_k,\qquad
    M_{k+1}=\mathsf{C}^{(Z)}_k(X_k^\dagger)N_{k+1}\mathsf{C}^{(Z)}_k(Y_k^\dagger).
\]
Then $M_k=\mathsf{C}^{(Z)}_k(L_k)M_{k+1}\mathsf{C}^{(Z)}_k(R_k)$ and
$\det L_k=\det R_k=1$. Both corrections fix $F_k\oplus E_k$,
so $M_{k+1}$ still fixes $F_{k+1}$ pointwise. Moreover,
$\det \mathsf{C}^{(Z)}_k(V)=\det V$, and hence $\det M_{k+1}=\det M_k=1$.
Starting with $M_1=H^{\otimes n}UH^{\otimes n}$, which fixes $F_1$,
and iterating gives the computational-basis factorization.

Conjugating by $H^{\otimes n}$ gives the claimed factorization, since
\[
    H^{\otimes n}\mathsf{C}^{(Z)}_i(V)H^{\otimes n}
    =\mathsf{C}_i\bigl(H^{\otimes(n-1)}VH^{\otimes(n-1)}\bigr),
\]
and conjugation preserves $\operatorname{SU}(2^{n-1})$.
\end{proof}

\section{Proof of~\cref{lem:adaptive-positive-representation}}
\label{app:adaptive-positive-representation}
Here, we prove the following restatement of~\cref{lem:adaptive-positive-representation}:
    Fix an adaptive $q$-query algorithm, where $1\leq q\leq d$.
There is a positive linear map $\Lambda$, depending only on the
algorithm, such that its averaged output state under an ensemble
$\xi$ is $\rho_\xi=\Lambda(J_\xi)$, where
\[
    J_\xi=d^{2q}\mathbb E_{U\sim\xi}
       (U_X^\dagger)^{\otimes q}|\Omega_{d^q}\rangle\langle\Omega_{d^q}|
       U_X^{\otimes q}.
\]
\begin{proof}[Proof of~\cref{lem:adaptive-positive-representation}]
Purify the query algorithm. At query $j$, retain its intended input
as $X_j$, introduce a fresh Bell pair $Y_jZ_j$ in
$|\Omega_d\rangle$, and continue the algorithm using $Z_j$ as the
oracle output. Retaining all $X_j,Y_j$ produces a normalized,
oracle-independent state $|\Psi\rangle_{WXY}$.
Gate teleportation gives
\[
    (\langle\Omega_d|_{XY}\otimes I_Z)
    (U_X\otimes I_{YZ})
    (|\chi\rangle_X\otimes|\Omega_d\rangle_{YZ})
    =d^{-1}U|\chi\rangle_Z,
\]
also when the input is entangled with a workspace. Contracting
all retained pairs thus reproduces the algorithm's output with
amplitude $d^{-q}$. Consequently,
\[
    \Lambda(Z)=\operatorname{Tr}_{XY}
       [(I_W\otimes Z)|\Psi\rangle\langle\Psi|]
\]
satisfies $\rho_\xi=\Lambda(J_\xi)$.
\end{proof}

Last, we show that $\rho_*=\Lambda(J_*)$ is a state:
First, $\Lambda$ is clearly positive and it remains to show that $\Lambda(J_*)$ is normalized.
Let $\omega_j$ be the reduction of $|\Psi\rangle\langle\Psi|$
to the first $j$ retained pairs. Causality gives
\[
    \omega_j=\tau_j\otimes I_{Y_j}/d,
    \qquad
    \operatorname{Tr}_{X_j}\tau_j=\omega_{j-1},
    \qquad\omega_0=1.
\]
Indeed, $Y_j$ is initially maximally mixed and independent of
the previously retained registers and $X_j$; later operations
act only on their complement.

Define
\[
    T_j=\sum_{\sigma\in S_j}\operatorname{Tr}
       [(R_{X_{\leq j}}(\sigma)\otimes R_{Y_{\leq j}}(\sigma))
        (I\otimes\Pi_{\mathrm{dist}(d,j)})\omega_j],
    \qquad T_0=1.
\]
When tracing out $Y_j$, a term with $\sigma(j)\neq j$ vanishes:
its diagonal entries would require two labels of a distinct tuple
to coincide. If $\sigma(j)=j$, there are $d-j+1$ choices for the
last label. The factor $I_{Y_j}/d$ therefore contributes
$(d-j+1)/d$. The surviving permutations act trivially on $X_j$,
so tracing $X_j$ replaces $\tau_j$ by $\omega_{j-1}$. Hence
\[
    T_j=\frac{d-j+1}{d}T_{j-1},
    \qquad
    \operatorname{Tr}\Lambda[J_0(I_X\otimes\Pi_Y)]
       =T_q=\frac{(d)_q}{d^q}.
\]
Thus $\operatorname{Tr}\Lambda(J_*)=1$.
\end{document}